\documentclass[11pt,letterpaper]{article}

\usepackage[english]{babel}

\usepackage{graphicx} % Required for inserting images
\usepackage[margin=1in]{geometry}
\usepackage[english]{babel}

\usepackage{xspace}
\usepackage{caption}
\usepackage{subcaption}
\usepackage{amssymb}
\usepackage{natbib}
\usepackage{amsmath}
\usepackage{amsthm}
\usepackage{graphicx}
\usepackage{multicol}
\usepackage{diagbox}
\usepackage{etoolbox}
\usepackage{pgfplots}
\pgfplotsset{compat=1.18}
\usepackage[colorlinks=true, allcolors=teal]{hyperref}
\usepackage{enumerate}    
\usepackage[ruled,linesnumbered,vlined]{algorithm2e}
\newtheorem{theorem}{Theorem}[section]

\newtheorem{lemma}[theorem]{Lemma}

\newtheorem{claim}[theorem]{Claim}
\newtheorem{definition}[theorem]{Definition}

\newtheorem{observation}[theorem]{Observation}

\usepackage{tikz}
\usepackage{siunitx}
\usepgfplotslibrary{fillbetween}
\usepackage{framed}
\definecolor{shadecolor}{gray}{0.85}
\usepackage{tcolorbox}
\tcbset{
  myframe/.style={
    colback=gray!5,
    colframe=black,
    boxrule=0.8pt,
    arc=4pt,
    outer arc=4pt,
    boxsep=5pt,
    left=5pt,
    right=5pt,
    top=5pt,
    bottom=5pt,
    enhanced,
    sharp corners
  }
}
\DeclareFontShape{OT1}{cmr}{m}{scit}{
    <-> ssub * cmr/m/sc % 替换为小体
}{}
\newenvironment{proofof}[1]{{\vspace*{5pt} \noindent\bf Proof of #1:  }}{\hfill\rule{2mm}{2mm}\vspace*{5pt}}

\newcommand{\ALG}{\mathsf{ALG}}
\newcommand{\OPT}{\mathsf{OPT}}
\newcommand{\E}{\mathbb{E}}
\newcommand{\LP}{\mathsf{LP}}

\newcommand{\bx}{\boldsymbol{x}}
\renewcommand{\d}{\mathrm{d}}

\title{Online Matching with KIID Edge Arrivals}
\author{Yilong Feng \thanks{State Key Lab of IOTSC, University of Macau. \{yc37459,yc47435,xiaoweiwu\}@um.edu.mo. The authors are ordered alphabetically. Xiaowei Wu is funded by the Science and Technology Development Fund (FDCT), Macau SAR (file no. 0147/2024/RIA2, 001/2024/SKL, 0002/2025/EQP and CG2026-IOTSC), and University of Macau (file no. MYRG-GRG2025-00033-IOTSC).}
	\and Haolong Li $^*$
        \and Xiaowei Wu $^*$}
\date{}

\begin{document}

\begin{titlepage}
\maketitle
\thispagestyle{empty}

\begin{abstract}
In the classic online stochastic matching proposed by Feldman et al. (FOCS 2009), there is a known bipartite type-graph, where one side of the graph is given offline.
Upon the arrival of each online vertex, its type is sampled independently and identically from the other side of the type-graph.
This model has been extensively studied over the past decade, yielding a rich body of theoretical results.
In this paper, we initiate the study of an edge arrival model for online stochastic matching. 
In our model, the online edges are sampled independently and identically (KIID) from a known type-graph, which need not be bipartite. 

We first show that the Greedy algorithm cannot achieve a competitive ratio strictly better than $0.5$ while the Suggested Matching algorithm has a competitive ratio of $1-1/e$ under the assumption of integral arrival rates, matching its performance in the one-sided vertex arrival model.
We then propose a two-stage algorithm that combines Greedy and Suggested Matching, and show that its competitive ratio is strictly higher than $1-1/e$ for integral arrival rates.
While our algorithm is simple, its analysis is intricate and builds upon the Natural LP, which has been proven very powerful in vertex arrival models.
Our result reveals that even in the more challenging edge arrival setting for general graphs, competitive ratios better than $1-1/e$ are still possible, given the known distributions.
\end{abstract}

\end{titlepage}

% Paper body
\section{Introduction}

Online matching, along with its variants, has been studied for decades since the seminal work of \citet{conf/stoc/KarpVV90}.
Consider an underlying bipartite graph where the vertices of one side are known upfront, while those on the other side arrive online in an arbitrary order.
Upon the arrival of an online vertex, the online algorithm needs to decide whether to match it immediately, to maximize the size of the final matching.
The problem is known as \emph{Online Bipartite Matching}, for which \citet{conf/stoc/KarpVV90} proposed an optimal $(1-1/e)$-competitive algorithm called \textsc{Ranking}.
It regained enormous attention in the early 2000s due to its explicit connection to online advertising~\citep{journals/fttcs/Mehta13}, where advertisers are modeled as offline vertices, while ad slots correspond to online vertices.
Nevertheless, the assumption of adversarial arrival order is often criticized for being too pessimistic.

\citet{conf/focs/FeldmanMMM09} introduced the \emph{Online Stochastic Matching} problem, where online vertices are sampled independently and identically from a known distribution represented by a type-graph.
This model aligns particularly well with the context of online advertising, as the expected number of ad slots (online vertices) can be predicted by historical data.
They proposed a $(1-1/e)$-competitive algorithm called \textsc{Suggested Matching} that makes matching decisions based on an offline solution.
They then improved upon this (beat $1-1/e$) under the assumption of \emph{integral arrival rates}, where the arrival rate (expected appearances in the realized graph) of each online vertex type is integral.
Without this assumption, \citet{journals/mor/ManshadiGS12}, who designed a $0.702$-competitive algorithm, were the first to surpass $1-1/e$.
The ratio was later improved by subsequent works \citep{journals/mor/JailletL14,conf/stoc/HuangS21}, resulting in the state-of-the-art ratio of $0.716$ by \citet{conf/stoc/0002SY22}.
The current best (smallest) upper bound is $1-\frac{e}{e^e}\approx 0.8206$, shown by \citet{journals/corr/ChierichettiGPV25} recently.

However, all the above results are established on the assumption that the vertices arrive online (and initiate the matching), which fails to capture scenarios where vertices are offline while matching opportunities arrive online.
For example, consider a rental market with tenants and landlords (vertices on different sides).
A rental opportunity (edge) lies between a tenant and a landlord if their requirements are compatible.
Since such opportunities usually appear after viewings or price negotiations, it would be more reasonable to assume edges (rather than vertices) arrive online.

In this paper, we initiate the study of online stochastic matching with edge arrivals.
Suppose that the vertices $V$ of the type-graph $G(V,E)$ are given offline.
There are $m$ edges that arrive online\footnote{In other words, the total arrival rate of all edge types in $E$ equals $m$.} and draw their types independently following a known and identical distribution on $E$.
Upon the arrival of each edge, the online algorithm must decide whether to match it, conditioned on both of its endpoints being unmatched, to maximize the total size of the resulting matching.
Note that since in the edge arrival model all vertices are given offline, the model naturally extends to general graphs.
Throughout this paper, we use $M^*$ to denote a maximum matching on $G$.
Let $|M^*| = n$.

\subsection{Discussion on Natural Algorithms}

We first discuss some natural algorithms for the problem, and measure their competitive ratio, which is the infimum of $\E[\ALG]/\E[\OPT]$ over all online instances, where $\ALG$ is the size of matching produced by the online algorithm and $\OPT$ is the maximum matching size of the realized graph.

We begin with the naive \textsc{Greedy} algorithm, which matches an edge as long as both its endpoints are unmatched at its arrival.
This algorithm is at least $0.5$-competitive, as it always produces a maximal matching in the realized graph. 
While it completely disregards the structure of the type-graph and is clearly suboptimal, it still achieves a competitive ratio of $1-1/e$ under the vertex arrival model~\citep{conf/soda/GoelM08}.

However, we show that its competitive ratio under the edge arrival model is at most $0.5$.
We demonstrate that \textsc{Greedy} matches only half of the vertices when the type-graph is a \emph{sunflower} — a graph with $2n$ vertices, where $n$ of them form a complete graph, and each of the remaining $n$ vertices has a unique neighbor in the complete graph. Intuitively, since most edges are sampled from the complete graph, the algorithm quickly matches nearly all vertices in the complete graph in the early stage of the algorithm, leaving most of the remaining vertices unmatched.
Furthermore, by replacing each vertex with multiple copies and each edge with a complete bipartite graph between the copies of its endpoints, we show that \textsc{Greedy} still matches only half of the vertices in the resulting graph, while nearly all vertices could be matched in the realized graph in expectation (see Section~\ref{section:Greedy} for a formal proof).

Another widely studied class of algorithms for online stochastic matching leverages a (fractional) solution to a linear program (LP) that characterizes the offline optimal solution. Several LP formulations have been proposed for the online stochastic matching problem, including the Basic LP, the Jaillet-Lu LP~\citep{journals/mor/JailletL14,conf/wine/QiuFZW23,conf/soda/Yan24,Yan25}, and the more powerful Natural LP~\citep{journals/mp/TorricoAT18,conf/stoc/HuangS21,journals/teco/MaXX22,conf/stoc/0002SY22}, whose optimal values serve as upper bounds on $\E[\OPT]$.
Consider the \textsc{Suggested Matching} algorithm, proposed by \citet{conf/focs/FeldmanMMM09}. Given a feasible LP solution $\{x_e\}_{e\in E}$, the algorithm proposes to match an online edge $e$ with probability proportional to $x_e$.
In the vertex arrival model, upon the arrival of an online vertex $v$, all its incident edges are revealed, and each offline vertex $u$ is proposed with probability $x_{(u,v)}/\text{arrival rate of } v$.
$v$ is then matched to $u$ if $u$ is proposed and still unmatched.
By lower bounding the probability that an offline vertex remains unmatched, it can be shown that the algorithm achieves a competitive ratio of $1 - 1/e$. Notably, this guarantee holds per edge and thus extends to edge-weighted graphs.

In the edge arrival model, a natural variant of \textsc{Suggested Matching} proceeds as follows:
\begin{itemize}
    \item Compute a feasible solution $\bx$ to an LP whose optimum upper bounds $\E[\OPT]$;
    \item Upon the arrival of an edge $e$, if both of its endpoints are unmatched, we match it with probability $x_e/\text{arrival rate of } e$.
\end{itemize}

To establish a competitive ratio of $\alpha$, it suffices to show that each edge is included in the matching with probability at least $\alpha \cdot x_e$ (a guarantee analogous to that in the Online Contention Resolution Scheme (OCRS)~\citep{conf/soda/FeldmanSZ16}). 
However, unlike the vertex arrival model, we now require \emph{both} endpoints to be unmatched, which makes it harder to derive a good lower bound on the competitive ratio:
By a similar analysis as in \citep{conf/focs/FeldmanMMM09}, one can get a $0.432$ per-edge guarantee for \textsc{Suggested Matching} in the edge arrival model (see Appendix~\ref{appendix-discussion}).
Even through a more careful analysis, or employing advanced online rounding techniques, it is difficult to derive a per-edge guarantee better than $0.5$ for general $\bx$, as suggested by existing works on Random-order OCRS~\citep{journals/mor/PollnerRSW24,conf/stoc/MacRuryM24,journals/ior/MacRuryMG25}.
Considering this inherent difficulty, a realistic target for the KIID edge arrival model would be obtaining a competitive ratio strictly better than $0.5$.

Inspired by existing works, we take the first step towards solving this problem by focusing on the case with integral arrival rates.
Surprisingly, we find that $1-1/e$ is in fact achievable, by feeding \textsc{Suggested Matching} with a \emph{well-structured} solution $\bx$.
Interestingly, rather than using solutions from powerful LPs (e.g., Jaillet-Lu or Natural LP), we find it more effective to use an indicator solution corresponding to a maximum matching $M^*$ of the type-graph, with $x_e = 1$ if $e \in M^*$, and $x_e = 0$ otherwise. 
That is, the algorithm only considers edges sampled from $M^*$, ignoring all others -- even if both endpoints are unmatched.
Since there is no contention among edges in $M^*$, it is straightforward to derive a competitive ratio of $1 - 1/e$ (see Section~\ref{section:suggested-matching} for a formal proof).

This leads to a natural question for the problem with integral arrival rates: 

\begin{quote}
    \emph{Can we achieve a competitive ratio strictly better than $1 - 1/e$?}
\end{quote}

Unfortunately, \textsc{Suggested Matching} cannot guarantee a ratio better than $1-1/e$:
Consider a type-graph being a complete graph on $2n$ vertices. In this case, it can be shown that $\mathbb{E}[\OPT] \approx n$, while \textsc{Suggested Matching} matches only $(1 - 1/e) \cdot n$ edges in expectation.

\paragraph{Remark:}
Notice that the problem differs significantly depending on whether the arrival rates are integral.
Take the edge-weighted online stochastic matching as an example. With the assumption of integral arrival rates, a competitive ratio of $0.667$ was derived by \citet{conf/wine/HaeuplerMZ11} (first to beat $1-1/e$); this was later improved to $0.705$ by \citet{journals/algorithmica/BrubachSSX20a}.
Without the assumption of integral arrival rates, however, progress has been more challenging: \citet{conf/stoc/0002SY22} established an upper bound (on the competitive ratio) of $0.703$, proving that the problem is strictly harder.
The $1-1/e$ ratio remained the state-of-the-art for more than a decade until \citet{conf/soda/Yan24} recently broke through with a $0.645$-competitive algorithm.

\subsection{Our Results}

In this work, we answer the above question affirmatively, by proposing the \textsc{Boosted Suggested Matching} algorithm, which combines \textsc{Suggested Matching} and \textsc{Greedy}.
We show that its competitive ratio is strictly larger than $1-1/e$ for the online matching with KIID edge arrivals under the assumption of integral arrival rates.

\paragraph{Main Idea of the Algorithm.}
Observe that the \emph{matching rate} of \textsc{Suggested Matching} is decreasing over time. 
After selecting $\beta\cdot n$ edges, the probability of an online edge being selected by the algorithm drops to $\frac{(1-\beta)n}{m}$, which is quite wasteful --- especially in the late stage of the algorithm when the type-graph is \emph{dense}, e.g., admits multiple maximum matchings.
In contrast, for such graphs, \textsc{Greedy} has a high matching rate, provided that a \emph{random} constant fraction of vertices remain unmatched.
Motivated by this insight, we boost the performance of \textsc{Suggested Matching} by switching to \textsc{Greedy} in the late stage of the algorithm.
Our approach is a two-stage algorithm parameterized by $\rho\in [0,1]$: we run \textsc{Suggested Matching} for the first $\rho\cdot m$ rounds, and switch to \textsc{Greedy} for the remaining rounds.

\begin{tcolorbox}[title=Boosted Suggested Matching]
% Initialize $M \gets \varnothing$. \\[0.2cm]
\textbf{For edge $e_k$ arriving in round $k \in \{ 1,2,\ldots,m \}$:}
\begin{enumerate}
  \item If $k \leq \rho\cdot m$: select the edge if $e_k \in M^*$ and both endpoints of $e_k$ are unmatched;
  \item If $k > \rho\cdot m$: select the edge if both endpoints of $e_k$ are unmatched.
\end{enumerate}
\end{tcolorbox}

% Note that our algorithm (parameterized by $\rho$) only relies on the maximum matching $M^*$ of the type-graph, and does not require the computation of any LP solutions.
%

\begin{tcolorbox}[colback = gray!15, colframe = gray!15]
{\bf Main Theorem.~}
The \textsc{Boosted Suggested Matching} algorithm achieves a competitive ratio of at least $0.6342$ for Online Matching with KIID Edge Arrivals.
\end{tcolorbox}

We further show that if the type-graph of the instance admits a perfect matching (which is without loss of generality in most vertex arrival online matching problems), then the competitive ratio of the algorithm improves to at least $0.6383$.
Finally, we complement our algorithmic result with a negative result showing that no online algorithm can achieve competitive ratios strictly better than $0.845$ under the edge arrival model.

\paragraph{Remark 1:}
The edge arrival model is the most general model for online matching, and is proposed as an important open problem by \citet{journals/fttcs/Mehta13}.
Yet it has not received sufficient attention due to a strong hardness result by \citet{conf/focs/GamlathKMSW19}, who showed that no algorithm, even randomized, can achieve a competitive ratio strictly higher than $0.5$.
To circumvent this challenge, existing works have explored alternative models like the random arrival model~\citep{conf/ipco/GuruganeshS17} and the stochastic realization model~\citep{conf/icalp/GravinTW21}, though most of these approaches yield competitive ratios only slightly above the $0.5$ threshold.
Our work proposes a more effective approach by leveraging known distributions, showing that significantly better competitive ratios can be achieved under this natural model. 
This is particularly important for practical online applications, as the adversarial arrival assumptions are often criticized for being overly pessimistic. 
Besides, our analysis might shed lights towards solving other edge arrival models.

\paragraph{Remark 2:}
Our algorithm falls into the class of two-stage (or multi-stage) algorithms, which has recently been applied to several online problems, e.g., edge-weighted online stochastic matching~\citep{conf/soda/Yan24,conf/wine/QiuFZW23,Yan25}, prophet inequality~\citep{journals/ai/FengLLWW25,journals/mor/Perez-SalazarST25}, and prophet matching~\citep{conf/soda/Chen0LT25}.

\subsection{Our Techniques}
\label{section:our-techniques}

As is typical in the analysis of two-stage algorithms, the main challenge lies in demonstrating that the more aggressive approach adopted in the second phase is beneficial, as opposed to the conservative strategy in the first phase.
In our analysis, we focus on measuring the probability of matching the online edge in each round.
Formally, we define the matching rate $r_k$ as $\frac{m}{n}$ times this probability in round $k$.
The competitive ratio can then be estimated by $\frac{1}{m}\cdot \sum_{k=1}^m r_k$, the average matching rate across all rounds.
It can be shown that for $k \leq \rho \cdot m$, the matching rate of \textsc{Boosted Suggested Matching} is $e^{-\frac{k}{m}}$, which decreases from $1$ to $e^{-\rho}$ as $k$ increases.
Ideally, if we can show that for $k > \rho \cdot m$, the matching rate $r_k$ (of \textsc{Greedy}) is always lower bounded by some constant $c$, then we can lower bound $\mathbb{E}[\ALG]$ by $\left( 1 - e^{-\rho} + (1 - \rho) \cdot c \right) \cdot n$.
% \begin{equation*} 
% \mathbb{E}[\ALG] \geq \left( 1 - e^{-\rho} + (1 - \rho) \cdot c \right) \cdot n. 
% \end{equation*} 

As long as $c > 1/e$, the above expression (with an optimized $\rho$) yields a lower bound strictly better than $1 - 1/e$ for the competitive ratio of the algorithm (see Figure~\ref{figure:matching-rate} for an example).

\begin{figure} [htb]
    \centering
    \resizebox{0.6\textwidth}{!}{
    \begin{tikzpicture}
    \begin{axis}[
    width=14cm,
    height=8cm,
    ylabel={$\text{Matching Rate}$},
    xlabel={${k}/{m}$},
    grid=major,
    xmin=0,
    xmax=1,
    ymin=0,
    ymax=1.1,
    scaled x ticks=false,
    scaled y ticks=false,
    xticklabel style={/pgf/number format/fixed},
    yticklabel style={/pgf/number format/fixed},
    xtick={0,0.2,0.4,0.6,0.8,1.0},
    xticklabels={0,0.2,0.4,0.6,0.8,1.0},
    ytick={0.2,0.4,0.6,0.8,1.0},
    yticklabels={0.2,0.4,0.6,0.8,1.0},
    legend pos=north east,
    tick align=outside,
    tick pos=left,
    clip=false, ]

    \addplot[blue, line width=1.8pt, smooth, domain=0:1, name path=exp] {exp(-x)};

    \addplot[red, line width=1.8pt, domain=0.69315:1, name path=const] {0.5};

    \node[black, below] at (axis cs:0.69315, -0.02) {$\rho$};
    \node[black, left] at (axis cs:-0.02, 0.5) {$c$};

    \addplot[red, mark=*, mark size=2.5pt, only marks] coordinates {(0.69315, 0.5)};

    \draw[red, dashed, thin] (axis cs:0.69315, 0.5) -- (axis cs:0.69315, 0);
    \draw[red, dashed, thin] (axis cs:0.69315, 0.5) -- (axis cs:0, 0.5);
    \path[name path=xaxis1] (axis cs:0,0.003) -- (axis cs:0.69315,0.003);
    \path[name path=xaxis2] (axis cs:0.69315,0.003) -- (axis cs:1,0.003);

    \addplot[gray!20] fill between [
    of=exp and xaxis1,
    soft clip={domain=0.002:0.69315} ];

    \addplot[gray!20] fill between [
    of=const and xaxis2,
    soft clip={domain=0.69315:0.998} ];

    \end{axis}
    \end{tikzpicture} }
    \caption{Illustration of the matching rate in different rounds, where the blue line represents the rate of \textsc{Suggested Matching}, and the red line represents (the lower bound of) that of \textsc{Greedy}. The shaded area represents a lower bound for the competitive ratio.}
    \label{figure:matching-rate}
\end{figure}
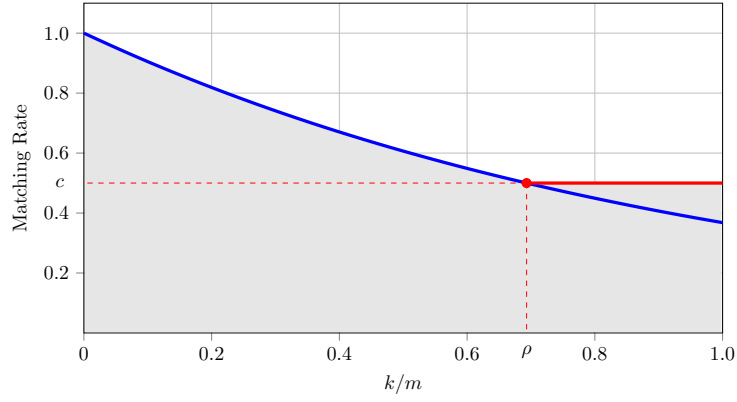

For example, if $c = 0.5$, then by setting $\rho = \ln(2) \approx 0.693$, we obtain a lower bound of $0.653$.
Thus, the main goal is to derive a constant $c > 1/e$ that lower bounds the matching rate in the second phase.
In general, this is not always possible: when the type-graph contains only $n$ edges, i.e., $E = M^*$, \textsc{Greedy} is equivalent to \textsc{Suggested Matching}, whose matching rate when $k \approx m$ is $1/e$.
However, in this case, our algorithm is already $1$-competitive since $\E[\OPT] = (1 - 1/e) \cdot n$.
Therefore, for breaking the $1-1/e$ barrier, it suffices to consider instances for which $\E[\OPT] \approx n$, which holds only when the type-graph has sufficiently many edges.
We show that for such instances, the matching rate of the second phase is lower bounded.

\paragraph{Lower Bounding Greedy.} 
Note that the algorithm we used in the second phase is a \emph{constrained} version of \textsc{Greedy}, with a fraction of the vertices ($1 - e^{-\rho}$ in expectation) being matched in the first phase.
Since directly lower bounding the matching rate in this constrained version is difficult, we use the unconstrained \textsc{Greedy} to approximate it and show that their matching sizes differ by a factor of at most $e^{-2\rho}$ (intuitively, as each vertex remains unmatched with probability $e^{-\rho}$ at the end of the first phase).
We then leverage the Natural LP, which upper bounds $\E[\OPT]$ and has an optimal value $\LP^* \approx n$.
Let $x_e$ be the variable associated with edge $e \in E$ in the LP. For every vertex $u$, we have the constraint of
\begin{equation} 
    \sum_{e \in S} x_e \leq 1 - e^{-|S|}, \qquad \forall S \subseteq E(u) \label{equation:LP_constraint} 
\end{equation} 
where $E(u)$ is the set of edges incident to $u$.
This constraint can be interpreted as: ``the contribution of edge set $S$ to matching $u$ is upper bounded by the probability that at least one of these edges arrives.''
Suppose the type-graph admits a perfect matching, i.e., $|V| = 2n$.
Given that $\LP^* \approx n$, we know that for almost all vertices $u$, we have $x_u = \sum_{e \in E(u)} x_e \approx 1$, implying that $u$ has many edges $e$ with non-zero $x_e$ by Constraint~\eqref{equation:LP_constraint}.
Therefore, if we run \textsc{Greedy} for $(1 - \rho) \cdot m$ rounds, e.g., $\rho = 0.99$, we can ensure that there are still many edges in the type-graph with both endpoints unmatched (which we call \emph{free edges}), resulting in a high matching rate.
Utilizing this observation, we show that \textsc{Boosted Suggested Matching} achieves a competitive ratio of at least $0.6383$ for type-graphs that admit a perfect matching.

\paragraph{Graphs without Perfect Matching.}
Unfortunately, in the edge arrival setting, it is no longer without loss of generality to assume that the type-graph admits a perfect matching.
This adds another layer of complexity to our analysis, as we can no longer argue that almost all vertices have $x_u \approx 1$ even if $\LP^* = n$ --- e.g., having $3n$ vertices with $x_u = 2/3$ also yields $\LP^* = n$.
% Indeed, it is easy to construct a sparse type-graph with $O(n)$ vertices and edges, for which the maximum matching $M^*$ has size $n$ and $\LP^* = n$.
% 
In this case, we need a special handling for vertices not matched in the maximum matching $M^*$.
Let $\hat{V} \subseteq V$ be the set of vertices matched in $M^*$.
We observe that such an instance must have sufficiently many, e.g., $\Theta(n)$, edges in the type-graph connecting a vertex in $\hat{V}$ to one not in $\hat{V}$.
Compared to edges connecting two vertices in $\hat{V}$, 
when using the unconstrained \textsc{Greedy} to lower bound the constrained version,
the contribution of edges between $\hat{V}$ and $V\setminus \hat{V}$ to the matching rate only needs to be scaled down by a factor of $e^{-\rho}$ (since vertices not in $\hat{V}$ cannot be matched in the first phase).
This allows us to derive a lower bound on the matching rate even when many edges exist between $\hat{V}$ and $V \setminus \hat{V}$, and show a lower bound of $0.6342$ on the competitive ratio of the algorithm.

\subsection{Other Related Works}

Due to the vast literature on online matching, in the following, we only review the most related works, e.g., on online stochastic matching and online matching with edge arrivals.
For a more comprehensive review of existing works, please refer to the recent survey by \citet{journals/sigecom/HuangTW24}.

\paragraph{Online Stochastic Matching with Integral Arrival Rates}
\citet{conf/focs/FeldmanMMM09} initially established a competitive ratio of $0.67$, and subsequent works~\citep{conf/esa/BahmaniK10,conf/wine/HaeuplerMZ11,journals/mor/ManshadiGS12,journals/mor/JailletL14,journals/algorithmica/BrubachSSX20a} gradually improved upon this.
The highest ratio is $0.7341$ given by \citet{xu2025}.
Meanwhile, the state-of-the-art upper bound with the integral arrival rates remains $0.8647$, as shown by \citet{journals/mor/ManshadiGS12}.

\paragraph{Online Stochastic Matching with Non-IID Vertex Arrivals.}
\citet{conf/stoc/TangWW22} generalized the online stochastic matching, allowing the distribution to be time-dependent.
They gave a $0.666$-competitive algorithm.
The ratio was recently improved to $0.69$ by \citet{conf/focs/Chen0S24}.
Note that both ratios also hold for the vertex-weighted case.
Online stochastic matching with correlated vertex arrivals was also studied by \citet{conf/sigecom/AouadM23}.

\paragraph{Weighted Online Stochastic Matching.}
The online stochastic matching on weighted graphs has also been considered extensively. 
\citet{journals/mor/JailletL14} first considered the vertex-weighted case of the problem with integral arrival rates, where offline vertices may have different weights, and proposed a $0.725$-competitive algorithm.
\citet{journals/algorithmica/BrubachSSX20a} then improved the ratio to $0.7299$.
For the edge-weighted case, competitive ratios of $0.667$ and $0.705$ were given by \citet{conf/wine/HaeuplerMZ11} and \citet{journals/algorithmica/BrubachSSX20a}.
The problem with general arrival rates has also been studied, e.g., for the vertex-weighted case~\citep{conf/stoc/HuangS21,conf/stoc/TangWW22} and edge-weighted case~\citep{conf/soda/Yan24,conf/wine/QiuFZW23}. 
The state-of-the-art competitive ratios are $0.716$ by \citet{conf/stoc/0002SY22} for the vertex-weighted case and $0.662$ by \citet{Yan25} for the edge-weighted case.

\paragraph{Online Matching with Edge Arrivals.}
In the most general model of online matching, the edges of an underlying graph arrive online following an adversarial order.
\citet{conf/focs/GamlathKMSW19} proved that, without additional information, no algorithm can be better than $0.5$-competitive, even on bipartite graphs.
Strictly better competitive ratios can be achieved for relaxed settings, e.g., for batched arrivals~\citep{conf/ipco/LeeS17}, for forests~\citep{journals/algorithmica/BuchbinderST19,journals/corr/abs-2410-21727}, and for known stochastic edge realizations~\citep{conf/icalp/GravinTW21}\footnote{Although their model has a similar name to ours, we remark that these two models are quite different: in their model the edge arrival order is adversarial, but the realization probability of the edges are known.}.
\citet{conf/ipco/GuruganeshS17} studied online matching with random edge arrivals, and proposed $(0.5+\Omega(1))$-competitive algorithms, which are the first to beat half.
The problem on forests was studied by \citep{journals/rsa/DyerF91}, who showed a tight competitive ratio of $0.7690$ for \textsc{Greedy}.
For the edge-weighted setting, \citet{conf/sigecom/EzraFGT22} proposed an $1/4$-competitive algorithm, while an upper bound of $1/e$ is implied by \citep{jour/dynkin1963}.

\paragraph{Online Contention Resolution Schemes (OCRS) for Edge Arrivals.}
OCRS was introduced by \citet{conf/soda/FeldmanSZ16} as a general rounding technique for online problems where items $e$ arrive online with activation probabilities $x_e$.
A $\gamma$-selectable OCRS is a rounding scheme that selects each item $e$ with probability at least $\gamma\cdot x_e$ satisfying downward-closed feasibility constraints.
% Here, we only review the OCRS for the matching polytopes where items are edges.
\citet{journals/mor/EzraFGT22} designed a $0.337$-selectable OCRS for edge arrivals on general graphs.
The ratio was then improved to $0.344$ (on general graphs) and $0.346$ (on bipartite graphs) by \citet{journals/ior/MacRuryMG25}, and to $0.382$ by \citet{conf/sigecom/MaMN24} for general graphs with vanishing probabilities.
For random edge arrivals, \citet{conf/nips/BrubachGMS21} designed a $0.432$-selectable Random-order OCRS on general graphs, which was improved to $0.45$ (on general graphs) and $0.456$ (on bipartite graphs) by \citet{journals/mor/PollnerRSW24}.
The state-of-the-art ratio for Random-order OCRS on general graphs is $0.474$ while that on bipartite graphs is $0.476$, given by \citet{journals/ior/MacRuryMG25}.

\section{Preliminaries}
\label{section:preliminiaries}

We consider the online matching problem with known independent and identically distributed (KIID) edge arrivals.
The input is a general (not necessarily bipartite) type-graph $G(V,E)$, where each $e\in E$ represents a unique edge type.
All vertices of the graph are given offline, while the edges arrive online and sample their types from the type-graph independently.
For each edge type $e$, we use $m_e$ to denote its expected number of appearances in the realized graph, which is an integer.
The total number of online edges is $\sum_{e\in E} m_e = m$, where each online edge samples its type $e$ from $E$ with probability $m_e/m$.
Note that by replacing each $e\in E$ with $m_e$ parallel edges between the two endpoints, we can assume w.l.o.g. that $m_e = 1$ for all edges $e$.
% and the total number of online edges equals $m$.
Note that the type-graph could now be a multi-graph.\footnote{While there could be multiple edges between two vertices, we still consider them as different edges. To avoid ambiguity, throughout this paper, we will use $e$ (instead of $(u,v)$) to denote an edge (between $u$ and $v$), unless it is clear that the graph is simple.}
Each online edge samples its type uniformly at random from $E$ upon its arrival.
The online algorithm must then immediately decide whether to match this edge, subject to the constraint that both of its endpoints are unmatched.
The goal is to maximize the size of the resulting matching in the realized graph.
Note that each edge type $e\in E$ may realize more than once.
For convenience, we denote the edge that arrives at time step $k$ as $e_k$.

Throughout this paper, we use $\ALG$ to denote the number of edges matched by the online algorithm, and $\OPT$ to denote that of the maximum matching on the realized graph.
Note that both $\ALG$ and $\OPT$ are random variables that depend on the arrival of edges.
% If the online algorithm is randomized, then $\ALG$ is further dependent on the random decisions of the algorithm.
The competitive ratio of the online algorithm is defined as the infimum of ${\E[\ALG]}/{\E[\OPT]}$ over all possible type-graphs $G$.
Let $M^*$ be a maximum matching on the type-graph $G$.
Note that if there are multiple parallel edges between vertices $u$ and $v$, at most one of them could be included in $M^*$.
Throughout this paper, we use $n = |M^*|$ to denote the size of the maximum matching, which implies that $|V| \geq 2n$.
When the graph admits a perfect matching, we have $|V| = 2n$.
We assume that $n$ and $m$ are sufficiently large.
Note that $n$ naturally serves as an upper bound of $\E[\OPT]$, since any realized graph (after deletion of repeated edges) is a subgraph of $G$.
In this paper, we further consider a tighter upper bound of $\E[\OPT]$ by introducing an edge arrival version of the Natural LP.

For all $u\in V$, let $E(u)$ be the set of edges incident to $u$.

\paragraph{Natural LP}
The Natural LP was first formally proposed by \citet{conf/stoc/HuangS21} for online stochastic matching with vertex arrivals.
We adapt it here for the KIID edge arrival model as follows.
\begin{align}
    \text{Maximize} \quad & \sum_{e\in E} x_e, \notag \\
    \text{subject to} \quad & \sum_{e\in S} x_{e} \leq 1-e^{-|S|}, & \forall u\in V, \ \forall S\subseteq E(u), \label{constraint} \\
    & 0\leq x_e \leq 1, \quad & \forall e \in E. \notag
\end{align}

Note that although the Natural LP contains an exponential number of constraints, it can be solved in polynomial time using the separation oracle, see, e.g., \citep{conf/stoc/HuangS21}.

Throughout this paper, we use $\LP^*$ to denote the optimal objective of the Natural LP.
The following lemma shows that $\LP^*$ is an upper bound of $\E[\OPT]$, which is a common observation for online stochastic matching.
For completeness, we provide a proof in Appendix~\ref{appendix}.

\begin{lemma} \label{lemma:LP*>=E[OPT]}
    For online matching with KIID edge arrivals, we have $\LP^* \geq \E[\OPT]$.
\end{lemma}

Next, we introduce two simple algorithms: \textsc{Greedy} and \textsc{Suggested Matching}, and provide a (tight) analysis on their competitive ratios under the KIID edge arrival model.

\subsection{Greedy}  \label{section:Greedy}

We first analyze the performance of the \textsc{Greedy} algorithm (see the following algorithm).
Observe that \textsc{Greedy} always outputs a maximal matching on the realized graph and thus has a competitive ratio of at least $0.5$.
In the following, we show that there exists an instance for which the competitive ratio of \textsc{Greedy} is upper bounded by $0.5 + o(1)$, even on bipartite graphs.

\begin{tcolorbox}[title=Greedy, label=algorithm:greedy]
    \textbf{For edge $e_k$ arriving in round $k \in \{ 1,2,\ldots,m \}$}:
    \begin{itemize}
        \item[]
        Select $e_k$ if both endpoints of $e_k$ are unmatched.        
    \end{itemize}
\end{tcolorbox}

\paragraph{Hard Instance.}
Consider the following bipartite graph with $n$ nodes on each side, where $n$ is sufficiently large such that $t := \sqrt{n}$ is an even integer.
Let there be $2n/t = 2t$ groups of vertices $A_1,A_2,\ldots,A_t,B_1,B_2,\ldots,B_t$, where each group has size $t$.
For each $i\in \{1,2,\ldots,t\}$, let $A_i$ and $B_i$ form a complete bipartite graph (so that there are $t^2 =n$ edges between them).
Let $\bigcup_{i\leq t/2} A_i$ and $\bigcup_{i> t/2} A_i$ form a complete bipartite graph (so that there are $(n/2)^2$ edges between them).
In total, the instance has $2n$ vertices and $m = (n/2)^2 + n^{1.5}$ edges.

In the following, we show that \textsc{Greedy} matches $(0.5+o(1))\cdot n$ vertices in expectation, while $\E[\OPT] \geq (1-o(1))\cdot n$, yielding an upper bound of $0.5+o(1)$ on the competitive ratio of \textsc{Greedy}.

\begin{theorem}  \label{thm:greedy_ub}
    \textsc{Greedy} is at most $0.5$-competitive for online matching with KIID edge arrivals.
\end{theorem}
\begin{proof}
    Let $A = \bigcup_{i\leq t} A_i$ and $B = \bigcup_{i\leq t} B_i$.
    Note that $|A| = |B| = n$.\footnote{Note that $A$ and $B$ are not the two sides of a bipartite graph.}
    
    We first show that \textsc{Greedy} matches $(0.5+o(1))\cdot n$ vertices in expectation.
    
    Consider running \textsc{Greedy} for $m' = n^{1.2}$ rounds.
    Note that in expectation, among the $m'$ randomly sampled edges, $\frac{(n/2)^2}{m} \cdot m' = \Theta(n^{1.2})$ of them are between two vertices in $A$; $\frac{n^{1.5}}{m} \cdot m' = \Theta(n^{0.7})$ of them are between a vertex in $A$ and a vertex in $B$, where an edge might be sampled multiple times.
    Intuitively speaking, since only $o(n)$ edges are sampled between $A$ and $B$, their contribution to Greedy is negligible even if all of them get matched.
    It suffices to show that after running \textsc{Greedy} for $m'$ rounds, almost all vertices of $A$ are matched.
    
    We first consider a virtual version of \textsc{Greedy} that rejects all edges between $A$ and $B$, and analyze the number of vertices in $A$ matched in the first $m'$ rounds.
    Note that since in expectation, $o(n)$ edges are sampled between $A$ and $B$ in the first $m'$ rounds, compared with the virtual \textsc{Greedy}, the expected number of vertices in $A$ matched by the (actual) \textsc{Greedy} in the first $m'$ rounds only differs by $o(n)$.
    Let $g(k)$ be the number of edges matched by the virtual \textsc{Greedy} algorithm at the end of round $k$, where we define $g(0) = 0$.
    Note that $g(k)$ is a random variable depending on the random samples of edges.
    Conditioned on $g(k)$, we have
    \begin{equation*}
        \E[g(k+1) \mid g(k)] = g(k) + \frac{(n/2 - g(k))^2}{m},
    \end{equation*}
    since the virtual \textsc{Greedy} matches an edge if and only if it is between two unmatched vertices in $A$.
    Taking an expectation over $g(k)$, we have
    \begin{equation*}
        \E[g(k+1)] \geq \E[g(k)] + \frac{(n/2 - \E[g(k)])^2}{m},
    \end{equation*}
    where the inequality follows from Jensen's inequality and the convexity of function $f(x) = (1-x)^2$.

    Next we prove that $\E[g(k)] \geq \frac{n}{2}\cdot (1-\frac{2m}{kn+2m})$ by mathematical induction on $k$. The base case of $k=0$ holds trivially due to $g(0) = 0$.
    Suppose the statement holds for $k$, we have
    \begin{align*}
        \E[g(k+1)] & \geq \E[g(k)] + \frac{(n/2 - \E[g(k)])^2}{m}
        \geq \frac{n}{2}\cdot \left( 1-\frac{2m}{kn+2m} \right) + \frac{(n/2 \cdot \frac{2m}{kn+2m})^2}{m} \\
        & = \frac{n}{2}\cdot \left( 1 - \frac{2m}{kn+2m} + \frac{n\cdot 2m}{(kn+2m)^2} \right)
        \geq \frac{n}{2}\cdot \left( 1-\frac{2m}{(k+1)n+2m} \right),
    \end{align*}
    where the second inequality holds since the function $\E[g(k)] + \frac{(n/2 - \E[g(k)])^2}{m}$ is increasing in $\E[g(k)]$; the last inequality holds due to
    \begin{equation*}
        \frac{1}{kn+2m} - \frac{1}{(k+1)n+2m} = \frac{n}{(kn+2m)((k+1)n+2m)} \leq \frac{n}{(kn+2m)^2}.
    \end{equation*}

    Therefore we have
    \begin{equation*}
        \E[g(m')] \geq \frac{n}{2}\cdot \left( 1-\frac{2m}{m'\cdot n+2m} \right) = \frac{n}{2} - O(n^{0.8}),
    \end{equation*}
    which implies that almost all vertices in $A$ are matched by the virtual \textsc{Greedy} algorithm.
    Therefore, the (actual) \textsc{Greedy} matches $n - o(n)$ vertices in $A$ in the first $m'$ rounds.
    Since only $o(n)$ vertices in $B$ are matched and $o(n)$ vertices in $A$ are unmatched, the total number of edges matched by \textsc{Greedy} (after $m$ rounds) is upper bounded by $(0.5+o(1))\cdot n$.

    Next we show that $\E[\OPT] \geq (1-o(1))\cdot n$.
    Recall that $\OPT$ is the size of the maximum matching of the realized graph.
    We can lower bound $\OPT$ by constructing a matching (of large expected size) for the realized graph. 
    Consider running a \textsc{Greedy} algorithm that accepts only edges between $A$ and $B$.
    We show that the algorithm matches almost all vertices.

    Fix an arbitrary $i\in \{1,2,\ldots,t\}$. Recall that $A_i$ and $B_i$ form a complete bipartite graph, where $|A_i| = |B_i| = t = \sqrt{n}$.
    Moreover, there is no edge between $A_i$ and $B_j$, or $A_j$ and $B_i$, for all $j\neq i$.
    Similar to the above analysis, we let $g(k)$ denote the number of edges matched by \textsc{Greedy} between $A_i$ and $B_i$ in the first $k$ rounds, and show that
    \begin{equation*}
        \E[g(k+1)] \geq \E[g(k)] + \frac{(t - \E[g(k)])^2}{m}.
    \end{equation*}

    Again, we can show by mathematical induction that $\E[g(k)] \geq t\cdot (1-\frac{m}{k t+m})$. Therefore we have $\E[g(m)] \geq t - O(1)$, which implies that after $m$ rounds, \textsc{Greedy} matches $n-o(n)$ edges and we have $\E[\OPT] \geq (1-o(1))\cdot n$.
\end{proof}

\subsection{Suggested Matching} \label{section:suggested-matching}

In this section, we show that, similar to the vertex arrival setting on bipartite graphs~\citep{conf/focs/FeldmanMMM09}, a competitive ratio of $1-1/e$ is achievable under the guidance of an offline solution on the type-graph $G$.
Recall that $M^*$ denotes a maximum matching on $G$.
Restating the \textsc{Suggested Matching} algorithm in the KIID edge arrival model (and feeding it with the indicator solution defined by $M^*$), we present the following algorithm.
Note that an edge $e\in M^*$ may be sampled multiple times but can be accepted only once.

\begin{tcolorbox}[title=Suggested Matching, label=algorithm:suggested-matching]
    \textbf{For edge $e_k$ arriving in round $k \in \{ 1,2,\ldots,m \}$}:
    \begin{itemize}
        \item[]
        Select $e_k$ if $e_k \in M^*$ and both endpoints of $e_k$ are unmatched. 
    \end{itemize}
\end{tcolorbox}

\begin{theorem}
    Feeding with the indicator solution defined by $M^*$, \textsc{Suggested Matching} achieves a competitive ratio of $1-1/e+o(1)$ for online matching with KIID edge arrivals, and the ratio is tight\footnote{We will omit some $o(1)$ items for simplicity of presentation in subsequent sections, which does not affect the correctness of our results as $m$ and $n$ are assumed to be sufficiently large.}.
\end{theorem}
\begin{proof}
    Observe that each edge $e\in M^*$ will be matched by \textsc{Suggested Matching} if and only if it ever arrives, because no edges that share a common endpoint with $e$ will be selected.
    Therefore, we can lower bound $\E[\ALG]$ as follows.
    \begin{align*}
        \E[\ALG] & = \sum_{e\in M^*} \Pr[e \text{ arrives at least once}] = \sum_{e\in M^*} \left( 1-\left( 1-\frac{1}{m} \right)^m \right) \\
        & = (1-1/e+o(1))\cdot n
        \geq (1-1/e+o(1))\cdot \E[\OPT].
    \end{align*}
    
    We remark that the above analysis is tight.
    Consider the type-graph being a complete bipartite graph $G(U\cup V,E)$ with $|U|=|V|=n$ and $|E|=n^2$ where $n$ is sufficiently large.
    Notice that this is equivalent to the case of $t = n$ in the above analysis (for \textsc{Greedy}).
    Therefore we can lower bound $\OPT$ by the performance of \textsc{Greedy} and derive that $\E[\OPT]=(1-o(1))\cdot n$.
    Recall that we have $\E[\ALG] = (1-1/e+o(1))\cdot n$.
    Therefore, we finish the proof.
\end{proof}

The above tight competitive analysis indicates the impossibility of beating $1-1/e$ (in general) using \textsc{Suggested Matching}, as there exist type-graphs on which $\E[\OPT]\approx n$.
However, for instances that admit a gap between $\E[\OPT]$ and $n$, \textsc{Suggested Matching} has a competitive ratio strictly larger than $1-1/e$.
This allows us to only consider instances with $\LP^*\geq (1-\epsilon)\cdot n$.

\begin{observation}
    For any $\epsilon > 0$, if the type-graph $G$ satisfies $\LP^* < (1-\epsilon)\cdot n$, then \textsc{Suggested Matching} achieves a competitive ratio of $\frac{1-1/e}{1-\epsilon} > 1-1/e+\epsilon/2$.
\end{observation}

\section{Boosted Suggested Matching} \label{section:boosted-suggested-matching}

In this section, we propose the \textsc{Boosted Suggested Matching} algorithm and lower bound its competitive ratio.
Before its matching decisions, our algorithm first computes the optimal objective of the Natural LP.
The choice of the parameter $\rho\in [0,1]$ is then adaptively determined by comparing $\LP^*$ against $n$.
If $\LP^* < (1-\epsilon)\cdot n$ where $\epsilon>0$ is a small constant to be optimized later, we set $\rho=1$ for completeness, under which the algorithm degenerates to \textsc{Suggested Matching}.
Otherwise, $\rho$ is set to the optimized constant derived in our analysis.
With $\rho$ properly initialized, our algorithm proceeds as follows.

\begin{tcolorbox}[title=Boosted Suggested Matching, label=algorithm:boosted-suggested-matching]
\textbf{For edge $e_k$ arriving in round $k \in \{ 1,2,\ldots,m \}$:}
\begin{enumerate}
  \item If $k \leq \rho\cdot m$: select the edge if $e_k \in M^*$ and both endpoints are unmatched;
  \item If $k > \rho\cdot m$: select the edge if both endpoints of $e_k$ are unmatched.
\end{enumerate}
\end{tcolorbox}

\begin{theorem}
    There is a $0.6342$-competitive algorithm for online matching with KIID edge arrivals.
    Moreover, if the type-graph admits a perfect matching, then the ratio can be improved to $0.6383$.
\end{theorem}

The remainder of this section is devoted to proving the above theorem.

\subsection{Analysis Framework} \label{section:Analysis_framework}

We first describe our general analysis framework for lower bounding the competitive ratio of the algorithm.
Let $M_1$ and $M_2$ be the (random) sets of edges matched in the first phase (from round $1$ to round $\rho\cdot m$) and the second phase (from round $\rho\cdot m + 1$ to the end) of the algorithm, respectively.
Let $m_1=|M_1|$ and $m_2=|M_2|$.
By a similar argument as in Section~\ref{section:suggested-matching}, the expected number of matched edges in the first phase is
\begin{equation} \label{equation:E[m_1]}
    \E[m_1] = \sum_{e\in M^*} \left( 1- \left( 1-\frac{1}{m} \right)^{\rho\cdot m} \right)
    = \left( 1-e^{-\rho} \right)\cdot n.
\end{equation}

Next, we provide a lower bound for $\E[m_2]$.
Notice that the algorithm we use in the second phase is just \textsc{Greedy}, constrained by some of the vertices (those matched in $M_1$) have already been matched when the algorithm begins.
We show in Lemma~\ref{lemma:lower-bound-of-E[m2]} that $\E[m_2]$ can be lower bounded by a constant fraction of the number of edges matched by the ``unconstrained'' \textsc{Greedy} algorithm.
% However, some of the vertices (those matched in $M_1$) have already been matched when the algorithm begins.
% We show in Lemma~\ref{lemma:lower-bound-of-E[m2]} that the number of edges matched by this constrained version of \textsc{Greedy} can be lower bounded by a constant fraction of that by the unconstrained version.
% 
Recall that $|M^*|=n$.
We introduce some notations:
\begin{itemize}
    \item Let $\hat{V}$ be the set of the $2n$ vertices matched in $M^*$. Let $\bar{V} = V\setminus \hat{V}$ be the other vertices.
    \item Let $\hat{E}$ be the set of all edges between vertices in $\hat{V}$, and $\bar{E}=E\setminus \hat{E}$.
\end{itemize}

Notice that $M^*\subseteq \hat{E}$ and there is no edge between any two vertices in $\bar{V}$ (since otherwise, $M^*$ is not a maximum matching).

We analyze the number of edges matched by the (unconstrained) \textsc{Greedy} algorithm.
For any $k>0$, let $g(k)$ be the number of edges matched by \textsc{Greedy} after running the algorithm for $k$ rounds.
To give a more precise characterization, we further define
\begin{itemize}
    \item $\hat{g}(k)$ be the number of edges in $\hat{E}\setminus M^*$ matched by \textsc{Greedy};
    \item $\bar{g}(k)$ be the number of edges in $\bar{E}\cup M^*$ matched by \textsc{Greedy}.
\end{itemize}

Note that both $\hat{g}(k)$ and $\bar{g}(k)$ are random variables non-decreasing in $k$ and we have $g(k) = \hat{g}(k)+\bar{g}(k)$.
For convenience, we also define $\hat{g}(0)=0$ and $\bar{g}(0)=0$.
The following lemma states that we can lower bound the size of the constrained \textsc{Greedy} by a constant fraction of the size of the unconstrained \textsc{Greedy}.
Intuitively, for each edge $e$ that contribute to $g(k)$, if $e\in \hat{E}\setminus M^*$, then $e$ has a probability of $e^{-2\rho}$ for contributing to $M_2$, as each of its endpoints is unmatched in the first phase with probability $e^{-\rho}$, and the two events are independent; the probability becomes $e^{-\rho}$ if $e\in \bar{E}\cup M^*$, since either only one of its endpoints could be matched (if $e\in \bar{E}$), or the two matching events are positively correlated (if $e\in M^*$).
For continuity of presentation, we defer the proof of the lemma to Appendix~\ref{appendix}.

\begin{lemma}
\label{lemma:lower-bound-of-E[m2]}
    We have
    $\E[m_2] \geq e^{-2\rho}\cdot \E[\hat{g}((1-\rho)m)] + e^{-\rho}\cdot \E[\bar{g}((1-\rho)m)]$.
\end{lemma}

Combining Equation~\eqref{equation:E[m_1]} and Lemma~\ref{lemma:lower-bound-of-E[m2]}, we can lower bound  the expected size of the matching output by \textsc{Boosted Suggested Matching} as follows:
\begin{align*}
    \E[\ALG] &= \E[m_1]+\E[m_2] \geq (1-e^{-\rho})\cdot n + e^{-2\rho}\cdot \E[\hat{g}((1-\rho)m)] + e^{-\rho}\cdot \E[\bar{g}((1-\rho)m) \\
    &= (1-e^{-\rho})\cdot n + e^{-2\rho}\cdot \left( \E\left[ \hat{g}((1-\rho)m) + e^{\rho}\cdot \bar{g}((1-\rho)m) \right] \right).
\end{align*}

Define $w(k) = \hat{g}(k) + e^{\rho}\cdot \bar{g}(k)$.
The main focus of our analysis is to provide a lower bound for $\E[w((1-\rho)m)]$.
We interpret $w(k)$ as the ``total weight'' of the matching produced by running \textsc{Greedy} for $k$ rounds, where edges from $\hat{E} \setminus M^*$ carry a weight of $1$ while those from $\bar{E}$ and $M^*$ have weight $e^\rho$.
To show a lower bound for $\E[w((1-\rho)\cdot m)]$, we argue that for all $k < (1-\rho)\cdot m$, there are many \emph{free edges} in the type-graph at the end of round $k$, where an edge is free if both its endpoints have not been matched by \textsc{Greedy}.
Intuitively speaking, the more free edges there are, the more likely the next online edge will be accepted by \textsc{Greedy}.

\begin{definition} [Free Edges]
    For each $k\geq 0$, an edge $e\in E$ is free at the end of round $k$ if both of its endpoints are not matched by \textsc{Greedy}.
    Let $F(k)$ be the (random) set of free edges at the end of round $k$.
    We further define $\hat{F}(k) = F(k)\cap (\hat{E}\setminus M^*)$, $\bar{F}(k) = F(k)\cap (\bar{E}\cup M^*)$.
\end{definition}

% Now we characterize $w(k)$ by the number of free edges as follows.

Fix the realization of edges in the first $k$ rounds, we have 
\begin{align*}
    w(k+1) = \begin{cases}
        w(k), \quad &\text{if } e_{k+1}\notin F(k); \\
        w(k) + 1, \quad &\text{if } e_{k+1}\in \hat{F}(k); \\
        w(k) + e^\rho, \quad &\text{if } e_{k+1}\in \bar{F}(k).
    \end{cases}
\end{align*}

Recall that each edge $e\in E$ realizes independently with probability $1/m$, we have
\begin{equation*}
    \E_{k+1}[w(k+1)] = w(k) + \frac{|\hat{F}(k)|}{m} + \frac{|\bar{F}(k)|}{m} \cdot e^\rho ,
\end{equation*}
where the expectation is taken only over the randomness from the $(k+1)$-th round.

Reorder the above equality and take another expectation (over the randomness of the first $k$ rounds) on both sides, we have
\begin{equation}  \label{equation:relation-between-greedy-and-free-edges}
    \E[w(k+1)] - \E[w(k)] = \frac{\E[|\hat{F}(k)| + e^{\rho}\cdot |\bar{F}(k)|]}{m}.
\end{equation}

So far, we have established the relation between the increment in $\E[w(k)]$ and the total weight of free edges at the end of round $k$.
Since $|\hat{F}(k)|$ and $|\bar{F}(k)|$ are non-increasing in $k$, the above equality also implies that the increment in $\E[w(k)]$ is non-increasing in $k$: the more edges \textsc{Greedy} matches, the less likely the next online edge can be accepted.

We will describe in Section~\ref{section:with-perfect-matching} and~\ref{section:without-perfect-matching} how to characterize the total weight of free edges using the Natural LP (recall that we have assumed $\LP^* \geq (1-\epsilon)\cdot n$), and lower bound $\E[w((1-\rho)m)]$.
We first present a few technical claims, which will be useful in the later analysis.

\begin{claim} \label{claim:sum_of_1-e^d}
    Given any non-negative integers $\{ d_1,d_2,\ldots,d_\ell \}$, we have
    \begin{equation*}
        \sum_{i=1}^\ell \left( 1 - e^{-d_i} \right) \leq \ell\cdot \left( 1 - e^{-\frac{1}{\ell}\sum_{i=1}^\ell d_i} \right).
    \end{equation*}
\end{claim}

\begin{claim} \label{claim:convexity_of_f}
    For any constants $c_1 > 0$ and $c_2 >0$, the following function is convex in $x\in (0,c_1)$:
    \begin{equation*}
        f(x) = (c_1 - x)\cdot \ln\left( \frac{c_1 - x}{c_2 + x} \right).
    \end{equation*}
\end{claim}

We omit the proofs of the claims as they can be verified using basic calculus.

\subsection{Instances with Perfect Matching}
\label{section:with-perfect-matching}

In this section, we consider type-graphs $G$ that admit perfect matchings, and give a lower bound of $0.6383$ for the competitive ratio.
Note that with this assumption, we have $\bar{V} = \bar{E} = \varnothing$, which implies $\hat{E} = E$, and thus $\hat{F}(k) = F(k)\setminus M^*$ and $\bar{F}(k) = F(k)\cap M^*$.
% This enables us to concentrate on establishing a lower bound for $\E[w(k)]$.
% 
Recall that our main goal is to provide a lower bound for $|\hat{F}(k)| + e^\rho\cdot |\bar{F}(k)|$, which can be equivalently expressed as $|F(k)| + (e^\rho - 1)\cdot |\bar{F}(k)|$ since $\hat{F}(k) \cup \bar{F}(k) = F(k)$.
By definition, $g(k)$ edges are matched by \textsc{Greedy} in the first $k$ rounds, which implies that $2g(k)$ vertices are matched.
Thus we have a lower bound for $|\bar{F}(k)|$:
\begin{equation*}
    |\bar{F}(k)| \geq n - 2g(k).
\end{equation*}

In the following, we focus on lower bounding $|F(k)|$.

Let $\bx$ be the optimal solution to the Natural LP.
Intuitively speaking, since the type-graph satisfies $\LP^*\geq (1-\epsilon)\cdot n$, we have $x_u = \sum_{e\in E(u)} x_e \approx 1$ for almost all vertices.
On the other hand, Constraint~\eqref{constraint} of the LP places an upper bound of $1-e^{-|E(u)|}$ on $x_u$.
Therefore, for almost all vertices $u$ we have $|E(u)| \gg 1$, which enables us to show that in the early stage of \textsc{Greedy} (when only a small fraction of vertices are matched), the matching rate of \textsc{Greedy} is very high since there are many free edges.
Formalizing this idea, we give a lower bound of $\E[|F(k)|]$ with respect to (w.r.t.) the expected size of matched edges $\E[g(k)]$.

\begin{lemma} \label{lemma:lower-bound-of-expected-free-edges}
    For any $k\geq 0$, we have
    \begin{equation*}
    \E[|F(k)|] \geq (n-\E[g(k)]) \cdot \ln\left( \frac{n-\E[g(k)]}{\epsilon\cdot n+\E[g(k)]} \right).
\end{equation*}
\end{lemma}

Before proving Lemma~\ref{lemma:lower-bound-of-expected-free-edges}, we first give the following intermediate statement, which serves as a ``per realization'' version of Lemma~\ref{lemma:lower-bound-of-expected-free-edges}.

\begin{lemma} \label{lemma:lower-bound-of-free-edges}
    For any $k\geq 0$, we have
    \begin{equation} \label{equation:lower-bound-of-free-edges}
        |F(k)| \geq (n-g(k)) \cdot \ln\left( \frac{n-g(k)}{\epsilon\cdot n+g(k)} \right).
    \end{equation}
\end{lemma}
\begin{proof}
    Let $U(k)$ be the set of vertices that have not been matched by \textsc{Greedy} by the end of round $k$.
    % Recall that $g(k)$ edges are matched by \textsc{Greedy} in the first $k$ rounds.
    We have $|U(k)| = 2n-2g(k)$.
    For each vertex $u\in U(k)$, let $f_u = |E(u)\cap F(k)|$ be the number of free edges incident to $u$.
    Now let's consider the contribution of all free edges to $\LP^*$ under the optimal solution $\bx$ of the LP.
    On the one hand, we have
    \begin{equation} \label{equation:lower-bound-of-partial-LP*}
        \sum_{e \in F(k)} x_e = \LP^* - \sum_{e\notin F(k)} x_e \geq \LP^* - \sum_{u\notin U(k)}\sum_{e\in E(u)} x_e \geq (1-\epsilon)n-2g(k),
    \end{equation}
    where the second inequality holds because (all the incident edges of) each vertex $u\notin U(k)$ contributes at most $1$ to $\LP^*$.
    On the other hand, we have
    \begin{equation*}
        \sum_{e \in F(k)} x_e = \frac{1}{2} \sum_{u\in U(k)}\sum_{e \in E(u)\cap F(k)} x_e \leq \frac{1}{2} \sum_{u\in U(k)} \left( 1-e^{-f_u} \right),
    \end{equation*}
    where the equality holds since every free edge is incident to two unmatched vertices in $U(k)$, and the inequality holds due to Constraint~\eqref{constraint}.
    
    Recall that by the definition of $f_u$, we have $\sum_{u\in U(k)} f_u = 2\cdot |F(k)|$.
    % \begin{equation*}
    %     \sum_{u\in U(k)} f_u = 2\cdot |F(k)|.
    % \end{equation*}
    
    By Claim~\ref{claim:sum_of_1-e^d}, we have
    \begin{equation*}
        \sum_{u\in U(k)} \left( 1-e^{-f_u} \right) \leq |U(k)| \cdot \left( 1-e^{-\frac{2|F(k)|}{|U(k)|}} \right) = 2 (n-g(k)) \cdot \left( 1-e^{-\frac{|F(k)|}{n-g(k)}} \right),
    \end{equation*}
    which leads to
    \begin{equation} \label{equation:upper-bound-of-partial-LP*}
        \sum_{e \in F(k)} x_e \leq (n-g(k))\cdot \left( 1-e^{-\frac{|F(k)|}{n-g(k)}} \right).
    \end{equation}

    Combining Equation~\eqref{equation:lower-bound-of-partial-LP*} and \eqref{equation:upper-bound-of-partial-LP*}, we have
    \begin{equation*}
        |F(k)| \geq (n-g(k)) \cdot \ln\left( \frac{n-g(k)}{\epsilon\cdot n+g(k)} \right),
    \end{equation*}
    which finishes the proof.
\end{proof}

Now we can prove Lemma~\ref{lemma:lower-bound-of-expected-free-edges}.

\begin{proofof}{Lemma~\ref{lemma:lower-bound-of-expected-free-edges}}
    By Claim~\ref{claim:convexity_of_f}, the RHS $(n-g(k)) \cdot \ln\left( \frac{n-g(k)}{\epsilon\cdot n+g(k)} \right)$ of Equation~\eqref{equation:lower-bound-of-free-edges} is convex in $g(k)$.
    Therefore, by Lemma~\ref{lemma:lower-bound-of-free-edges} and Jensen's Inequality, we have
    \begin{equation*}
        \E[|F(k)] \geq \E\left[ (n-g(k)) \cdot \ln\left( \frac{n-g(k)}{\epsilon\cdot n+g(k)} \right) \right] \geq (n-\E[g(k)]) \cdot \ln\left( \frac{n-\E[g(k)]}{\epsilon\cdot n+\E[g(k)]} \right),
    \end{equation*}
    which finishes the proof.
\end{proofof}

Putting all lower bounds together, we obtain:
\begin{align}
    \E[w(k+1)] - \E[w(k)] & = \frac{\E[|F(k)|]+(e^{\rho}-1)\cdot \E[|\bar{F}(k)|]}{m} \nonumber\\
    & \geq \frac{n-\E[g(k)]}{m} \cdot \ln\left( \frac{n-\E[g(k)]}{\epsilon\cdot n+\E[g(k)]} \right) + \frac{e^{\rho}-1}{m}\cdot (n-2\E[g(k)]) \nonumber\\
    & \geq \frac{n-\E[w(k)]}{m} \cdot \ln\left( \frac{n-\E[w(k)]}{\epsilon\cdot n+\E[w(k)]} \right) + \frac{e^{\rho}-1}{m}\cdot (n-2\E[w(k)]),
\label{equation:lower-bound-on-increment-of-E[g(k)]}
\end{align}
where the last inequality holds since the function $b(z) = (1-z)\cdot \ln\left( \frac{1-z}{\epsilon + z} \right) + (e^{\rho}-1)\cdot (1-2z)$ is decreasing in $z$ (when $z \in (0, (1-\epsilon)/2)$) and $w(k) = \hat{g}(k) + e^\rho\cdot \bar{g}(k) \geq g(k)$.
Therefore, we can lower bound the increment of $\E[w(k)]$ as a function of $\E[w(k)]$.
We make use of this characterization to give a lower bound on $\E[w((1-\rho)\cdot m)]$.

\begin{lemma} \label{lemma:lower-bound-of-E[g]}
    Let $f(x)$ with $f(0)=0$ be a function of $x\in [0,1)$ such that $f'(x)=b(f(x))$, where
    \begin{equation*}
        b(z) = (1-z)\cdot \ln\left( \frac{1-z}{\epsilon + z} \right) + (e^{\rho}-1)\cdot (1-2z).
    \end{equation*}
    For any $k\geq 0$ where $\E[w(k)] < \frac{1-\epsilon}{2}\cdot n$, we have $\E[w(k)]\geq n\cdot f(\frac{k}{m})$.
\end{lemma}

Intuitively speaking, $b(z)$ is a continuous function that lower bounds the increment of $\E[w(k)]$ when ${\E[w(k)]}/{n} = z$.
Therefore $f({k}/{m})$ serves as a lower bound for ${\E[w(k)]}/{n}$.
We leave the proof of the above lemma to Appendix~\ref{appendix}.
Although we are not able to derive an analytic form for $f(x)$, its numerical form can be computed to a very high precision for any fixed $\epsilon$ with the assistance of a computer program (see Figure~\ref{figure:function-f}).
By the above discussion, for any $\rho\in [0,1]$ that satisfies $\E[w((1-\rho)m)] \leq \frac{1-\epsilon}{2}\cdot n$, we can lower bound $\E[\ALG]$ as
\begin{align*}
    \E[\ALG] & \geq (1-e^{-\rho})\cdot n + e^{-2\rho} \cdot \E[w((1-\rho)m)]
    \geq (1-e^{-\rho})\cdot n + e^{-2\rho} \cdot n\cdot f(1-\rho) \\
    & \geq \left( 1-e^{-\rho} + e^{-2\rho} \cdot f(1-\rho) \right)\cdot \E[\OPT].
\end{align*}

\begin{figure}[htbp]
\centering
\resizebox{0.6\textwidth}{!}{
\begin{tikzpicture}
\begin{axis}[
    width=12cm,
    height=8cm,
    xlabel={$x$},
    ylabel={$f(x)$},
    grid=major,
    legend pos=north west,
    xmin=0,              
    ymin=0,              
    xmax=0.1,
    scaled x ticks=false, 
    scaled y ticks=false, 
    xticklabel style={/pgf/number format/fixed}, 
    yticklabel style={/pgf/number format/fixed}, 
    xtick={0,0.02,0.04,0.06,0.08,0.1},
    ytick={0.05,0.1,0.15,0.2,0.25}
]

\addplot[blue, line width=1.5pt] table {data/solution_data.txt};
\addplot[red, mark=*, mark size=2pt] coordinates {(0.05, 0.16884)}; 

\node[] at (axis cs:0.063,0.15) {$(0.05, 0.168)$};
\draw[black, dashed, thin] (axis cs:0.05,0.16884) -- (axis cs:0.05,0);
\draw[black, dashed, thin] (axis cs:0.05,0.16884) -- (axis cs:0,0.16884);

\end{axis}
\end{tikzpicture}
}
\caption{Function $f$ when $\epsilon = 0.01$.}
\label{figure:function-f}
\end{figure}
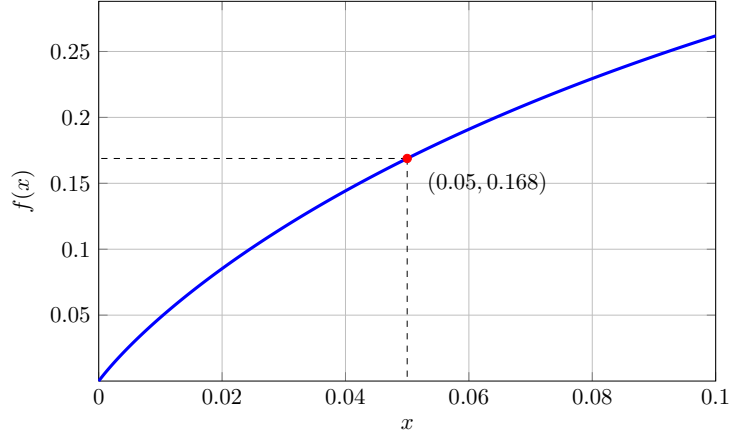

\begin{theorem}
    For instances whose type-graph admits perfect matchings, \textsc{Boosted Suggested Matching} is at least $0.6383$-competitive.
\end{theorem}
\begin{proof}
    Fix $\epsilon = 0.01$.
    If the type-graph $G$ satisfies $\LP^* < (1-\epsilon)\cdot n$, then (by setting $\rho=1$) the competitive ratio is at least
        \begin{equation*}
            \frac{1-1/e}{1-\epsilon} = \frac{1-1/e}{1-0.01} > 0.6383.
        \end{equation*}

    Otherwise, we set $\rho=0.95$.
    If $\E[w((1-\rho)m)] > \frac{1-\epsilon}{2}\cdot n$, then the competitive ratio is at least
    \begin{equation*}
        1-e^{-\rho} + e^{-2\rho} \cdot \frac{1-\epsilon}{2} = 1-e^{-0.95} + e^{-1.9} \cdot \frac{1-0.01}{2} > 0.6872 > 0.6383.
    \end{equation*}
    If $\E[w((1-\rho)m)] \leq \frac{1-\epsilon}{2}\cdot n$, then the competitive ratio is at least
    \begin{equation*}
        1-e^{-\rho} + e^{-2\rho} \cdot f(1-\rho) = 1-e^{-0.95} + e^{-1.9} \cdot f(0.05) > 0.6383. \qedhere
    \end{equation*}
\end{proof}

\subsection{General Instances}
\label{section:without-perfect-matching}

In this section, we consider general instances, which may not admit perfect matchings.
As before, our goal is to lower bound $\E[ w((1-\rho)m)]$, where $w(k) = \hat{g}(k) + e^\rho\cdot \bar{g}(k)$.
We define $\Delta(k) = |\hat{F}(k)| + e^{\rho}\cdot |\bar{F}(k)|$, the total weight of the free edges at the end of round $k$.
Note that both $w(k)$ and $\Delta(k)$ are random variables.
Recall that by Equation~\eqref{equation:relation-between-greedy-and-free-edges} we have
\begin{equation*}
        \E[w(k+1)] - \E[w(k)] = \frac{\E[\Delta(k)]}{m}.
\end{equation*}

Similar to the case with perfect matching, we would like to lower bound $\E[\Delta(k)]$ for each round $k$, which measures the expected increase of $w(k)$.
The following lemma is crucial to our analysis, which states that when $w(k)$ is small, $\Delta(k)$ must be large.

\begin{lemma} \label{lemma:small_w_large_delta}
   If $w(k) \leq \epsilon'\cdot n$, then $\Delta(k) \geq c(\rho,\epsilon,\epsilon')\cdot n$, where
    \begin{equation*}
        c(\rho,\epsilon,\epsilon') = e^\rho + 1+\ln\left( \frac{2e^{\rho}-1}{e-1} \right) - \frac{e(2e^{\rho}-1)}{e-1}\cdot \epsilon - \left( 3e^\rho-2+\frac{2e^\rho-1}{e-1} \right) \cdot 2\epsilon'.
    \end{equation*}
\end{lemma}

The lemma states that when we run \textsc{Greedy} for a short period, i.e., $\epsilon' \approx 0, \rho \approx 1$, then the matching rate (measured by $e^{-2}\cdot c(\rho,\epsilon,\epsilon')$) is close to
\begin{equation*}
    e^{-2}\cdot \left( e+1 + \ln\left( \frac{2e-1}{e-1} \right) \right) \approx \frac{4.667}{e^2} > \frac{1}{e}.
\end{equation*}

Recall (from Section~\ref{section:our-techniques}) that as long as we can show a lower bound strictly larger than $1/e$ for the matching rate in the second phase, we can derive a competitive ratio strictly larger than $1-1/e$.
For continuity of presentation, we defer the proof of the lemma to the end of this section, and derive a lower bound for the competitive ratio first.

% We first use Lemma~\ref{lemma:small_w_large_delta} to show that $\E[w(k)]$ reaches $\epsilon'\cdot n$ quickly.
Using Lemma~\ref{lemma:small_w_large_delta}, we further give Lemma~\ref{lemma:large_expected_g}, which states that $\E[w(k)]$ reaches $\epsilon'\cdot n$ quickly (by measure concentration).
We defer the proof to Appendix~\ref{appendix}.

\begin{lemma} \label{lemma:large_expected_g}
    For any $\rho$, $\epsilon$, and $\epsilon'$ such that $c(\rho,\epsilon,\epsilon') > 0$, we have $\E[w(k^*)] \geq (1-o(1)) \cdot \epsilon'\cdot n$, where $k^* = {\epsilon' \cdot m}/{c(\rho,\epsilon,\epsilon')}$.
\end{lemma}

Using the above lemma, we can establish a lower bound of $0.6342$.
\begin{theorem}
    \textsc{Boosted Suggested Matching} is at least $0.6342$-competitive for online matching with KIID edge arrivals.
\end{theorem}
\begin{proof}
    Fix $\epsilon = 0.0034$.
    If the type-graph satisfies $\LP^* < (1-\epsilon)\cdot n$, then (by setting $\rho=1$) the competitive ratio is at least
    \begin{equation*}
        \frac{1-1/e}{1-\epsilon} = \frac{1-1/e}{1-0.0034} > 0.63427.
    \end{equation*}
    
    Otherwise, we set $\rho=0.98$ and let $\epsilon' = 0.068$.
    It can be verified that
    \begin{equation*}
        c(\rho,\epsilon,\epsilon') > 3.4074 \quad \text{ and } \quad (1-\rho)\cdot m \geq k^* = \frac{\epsilon'\cdot m}{c(\rho,\epsilon,\epsilon')}.
    \end{equation*}    

    Therefore, we can lower bound the competitive ratio by
    \begin{align*}
         & \ 1-e^{-\rho} + e^{-2\rho} \cdot \frac{\E[w((1-\rho)m)]}{n}
         \geq \ 1-e^{-\rho} + e^{-2\rho} \cdot \frac{\E[w(k^*)]}{n} \\
         = & \ 1-e^{-\rho} + e^{-2\rho} \cdot (1-o(1))\cdot \epsilon' > 0.63426,
    \end{align*}
    where the inequality follows from the monotonicity of $\E[w(k)]$.
\end{proof}

It remains to prove Lemma~\ref{lemma:small_w_large_delta}.
Similar to the analysis in Section~\ref{section:with-perfect-matching}, we fix the optimal solution $\bx$ to the Natural LP, and show that there should be a sufficient number of (free) edges so that their total contribution to $\LP^*$ enables $\LP^* \geq (1-\epsilon)\cdot n$.
In the following, we fix some round $k$ and consider the contributions to $\LP^*$ from the following three parts: 
\begin{itemize}
    \item edges incident to vertices that are already matched;
    \item free edges in $\hat{E}$, i.e., edges in $F(k)\cap \hat{E} = \hat{F}(k)\cup (\bar{F}(k)\cap M^*)$;
    \item free edges in $\bar{E}$, i.e., edges in $F(k)\cap \bar{E} = \bar{F}(k) \setminus M^*$.
\end{itemize}

For convenience, we index the vertices $\hat{V}$ as $\{u_1,u_2,\ldots,u_{2n}\}$.
To characterize the contributions from the free edges in compliance with Constraint~\eqref{constraint}, we give the following definition.

\begin{definition}
    For each vertex $u_i$ where $i\in \{1,\ldots,2n\}$, let $d_i$ and $\delta_i$ be the number of its incident free edges in $\hat{E}$ and $\bar{E}$, respectively.
\end{definition}

Note that if a vertex $u_i$ has already been matched by \textsc{Greedy} (at or before round $k$), then we have $d_i = \delta_i = 0$.
According to the above definition, we have
\begin{equation*}
    |F(k)\cap \hat{E}| = \frac{1}{2} \sum_{i=1}^{2n} d_i \quad \text{and} \quad |F(k)\cap \bar{E}| = \sum_{i=1}^{2n} \delta_i.
\end{equation*}

\begin{lemma} \label{lemma:contribution-from-free-edges}
    The contribution to $\LP^*$ from all free edges at the end of round $k$ is at most
    \begin{equation*}
        n \left( 1-\exp\left(-\frac{2|F(k)\cap \hat{E}|+|F(k)\cap \bar{E}|}{2n}\right) \right) + \frac{(1-1/e)\cdot |F(k)\cap \bar{E}|}{2}.
    \end{equation*}
\end{lemma}
\begin{proof}
    Let $\{x_e\}_{e\in E}$ be the optimal solution to the Natural LP.
    We first consider the contribution of the free edges to $\LP^*$.
    For edges in $F(k)\cap \hat{E}$, we have
    \begin{align*}
        \sum_{e\in F(k)\cap \hat{E}} x_e 
        & = \frac{1}{2}\cdot \sum_{u\in \hat{V}} \sum_{e\in E(u)\cap F(k)\cap \hat{E}} x_e 
        = \frac{1}{2}\cdot \left( \sum_{u\in \hat{V}} \sum_{e\in E(u)\cap F(k)} x_e - \sum_{e\in F(k)\cap \bar{E}} x_e \right) \\
        & \leq \frac{1}{2}\cdot \left( \sum_{i=1}^{2n} (1-\exp(-d_i-\delta_i)) - \sum_{e\in F(k)\cap \bar{E}} x_e \right),
    \end{align*}
    where the first equality holds since each edge in $F(k)\cap \hat{E}$ is incident to two vertices in $\hat{V}$ and the inequality holds due to the feasibility of $\bx$ to the LP.
    Taking the contribution of the free edges in $F(k)\cap \bar{E}$ into account, the contribution from all free edges is at most
    \begin{align*}
        & \ \frac{1}{2}\cdot \left( \sum_{i=1}^{2n} (1-\exp(-d_i-\delta_i)) - \sum_{e\in F(k)\cap \bar{E}} x_e \right) + \sum_{e\in F(k)\cap \bar{E}} x_e \\
        = & \ \frac{1}{2}\cdot \left( \sum_{i=1}^{2n} (1-\exp(-d_i-\delta_i)) + \sum_{e\in F(k)\cap \bar{E}} x_e \right) \\
        \leq & \ n\cdot \left( 1-\exp\left(-\frac{2|F(k)\cap \hat{E}| + |F(k)\cap \bar{E}|}{2n}\right) + \left( 1-\frac{1}{e} \right)\cdot |F(k)\cap \bar{E}| \right),
    \end{align*}
    where the inequality holds due to Claim~\ref{claim:sum_of_1-e^d} and $x_e \leq 1-1/e$ for all edge $e$.    
\end{proof}

By feasibility of $\bx$, each matched vertex has a contribution of at most $1$ to $\LP^*$.
Note that there are $2g(k)$ matched vertices at the end of round $k$.
Together with Lemma~\ref{lemma:contribution-from-free-edges}, we can upper bound their total contribution to $\LP^*$ by
\begin{equation} \label{eq:constraint-of-growingspeed}
    2 g(k) + n \cdot \left( 1-\exp\left(-\frac{2|F(k)\cap \hat{E}| + |F(k)\cap \bar{E}|}{2n}\right) \right) + \frac{(1-1/e)}{2}\cdot |F(k)\cap \bar{E}|.
\end{equation}

Now we are ready to prove Lemma~\ref{lemma:small_w_large_delta}.

\begin{proofof}{Lemma~\ref{lemma:small_w_large_delta}}
    For ease of notation, we let
    \begin{equation*}
        \alpha = \frac{2|F(k)\cap \hat{E}| + |F(k)\cap \bar{E}|}{2n}, \quad \text{ and } \quad
        \beta = \frac{|F(k)\cap \bar{E}|}{n}.
    \end{equation*}

    Recall that $\LP^* \geq (1-\epsilon)\cdot n$.
    By Equation~\eqref{eq:constraint-of-growingspeed} we have
    \begin{align}
        (1-\epsilon)\cdot n & \leq 2 g(k) + n \cdot \left( 1-e^{-\alpha} \right) + \frac{(1-1/e)}{2}\cdot \beta\cdot n \nonumber \\
        & \leq 2\cdot \epsilon'\cdot n + n \cdot \left( 1-e^{-\alpha} \right) + \frac{(1-1/e)}{2}\cdot \beta\cdot n, \label{equation:total_contribution>1-eps}
    \end{align}
    where the second inequality holds due to $g(k) \leq w(k)$ and the assumption of the lemma that $w(k)\leq \epsilon'\cdot n$.
    Reordering the inequality, we obtain
    \begin{equation}
        \beta \geq \frac{e^{-\alpha}-(\epsilon+2\epsilon')}{\frac{1}{2}\cdot (1-\frac{1}{e})}. \label{equation:lower_bound_beta_by_alpha}
    \end{equation}

    On the other hand, we can rewrite $\Delta(k)$ as
    \begin{align*}
        \Delta(k) & = |\hat{F}(k)| + e^{\rho}\cdot |\bar{F}(k)| \\
        & = |F(k)\cap \hat{E}| - |F(k)\cap M^*| + e^{\rho}\cdot \left( |F(k)\cap \bar{E}| + |F(k)\cap M^*| \right) \\
        & = \left(\alpha + \left( e^{\rho}-\frac{1}{2} \right)\cdot \beta \right)\cdot n + \left( e^\rho - 1 \right)\cdot |F(k)\cap M^*| \\
        & \geq \left(\alpha + \left( e^{\rho}-\frac{1}{2} \right)\cdot \beta + (e^{\rho}-1)\cdot (1-2\epsilon')\right)\cdot n,
    \end{align*}
    where the inequality holds since the number of matched vertices is $2g(k)\leq 2w(k)\leq 2\epsilon'\cdot n$ at the end of round $k$.
    Plugging in the lower bound for $\beta$ (Equation~\eqref{equation:lower_bound_beta_by_alpha}), we obtain
    \begin{align*}
        \frac{\Delta(k)}{n} & \geq \alpha + \frac{e^{\rho}-\frac{1}{2}}{\frac{1}{2}\cdot (1-\frac{1}{e})}\cdot \Big( e^{-\alpha}-(\epsilon+2\epsilon') \Big) +(e^{\rho}-1)\cdot (1-2\epsilon')\\
        & = \alpha + \frac{e(2e^{\rho}-1)}{e-1}\cdot e^{-\alpha} - \frac{e(2e^{\rho}-1)}{e-1}\cdot (\epsilon + 2\epsilon') + (e^{\rho}-1)\cdot (1-2\epsilon').
    \end{align*}

    Finally, we observe that $\alpha + \frac{e(2e^{\rho}-1)}{e-1}\cdot e^{-\alpha}$ is convex in $(0,\infty)$ and achieves its minimum at $\alpha = \ln\left( \frac{e(2e^{\rho}-1)}{e-1} \right)$.
    Thus we have
    \begin{align*}
        \Delta(k) & \geq 1 + \ln\left( \frac{e(2e^{\rho}-1)}{e-1} \right) - \frac{e(2e^{\rho}-1)}{e-1}\cdot (\epsilon + 2\epsilon')+(e^{\rho}-1)\cdot (1-2\epsilon') \\
        & = e^\rho + 1+\ln\left( \frac{2e^{\rho}-1}{e-1} \right) - \frac{e(2e^{\rho}-1)}{e-1}\cdot \epsilon - \left( 3e^\rho-2+\frac{2e^\rho-1}{e-1} \right) \cdot 2\epsilon',
    \end{align*}
    which finishes the proof.
\end{proofof}

\section{Hardness Results}
\label{section:hardness}

In this section, we present an upper bound of $0.845$ for all algorithms for online matching with KIID edge arrivals, where in the hard instance, the type-graph $G$ is a complete graph with $4$ vertices.
% We will demonstrate that the optimal online algorithm has a competitive ratio of $0.8450$ on $G$.
% One can derive another type-graph with the same upper bound of $0.8450$ while the number of edges is sufficiently large, by copying $G$ many times.

\begin{theorem} \label{theorem:hardness}
    No online algorithm can be better than $0.845$-competitive for online matching with KIID edge arrivals.
\end{theorem}
\begin{proof}
    We first characterize the behavior of the optimal online algorithm for the stated hard instance (which can be straightforwardly extended to analyze other hard instances).
    
    Let $\texttt{DP}[M,k]$ be the expected number of edges matched (excluding $M$) by the optimal online algorithm conditioned on edges in $M$ being selected already and there are $k$ rounds left.
    We have $\E[\ALG] \leq \texttt{DP}[\varnothing,6]$ for any online algorithms.
    Notice that since the instance has only four vertices, $|M|$ can only be $0,1$ or $2$.
    When $|M| = 0$, the optimal online algorithm is free to match any arrived edges.
    When $|M| = 1$, the algorithm has to wait for the edge type with no common endpoints with that in $M$.
    When $|M| = 2$, the algorithm can not select any edge anymore.
    Since each edge type $e\in E$ realizes with probability $1/6$ at each round, we have
    \begin{align*}
        \texttt{DP}[M,k] = 
        \begin{cases}
            \frac{1}{6}\cdot \sum_{e\in E} \max\{1+\texttt{DP}[\{e\}, k-1], \texttt{DP}[\varnothing, k-1]\}, & \text{ if } |M| = 0 \text{ and } k\geq 1; \\
            \frac{1}{6} + \frac{5}{6}\cdot \texttt{DP}[M, k-1], & \text{ if } |M| = 1 \text{ and } k\geq 1; \\
            0, & \text{ if } |M| = 2 \text{ or } k = 0.
        \end{cases}
    \end{align*}
    
    By computing $\texttt{DP}[M,k]$ for all $M$ and $k$, one can verify that the optimal online algorithm always matches the first arrived edge $e_1$.
    Therefore, we can upper bound $\E[\ALG]$ as
    \begin{equation*}
        1 + \texttt{DP}[\{e_1\},5] = 1 + 1\cdot \left( 1-(5/6)^5 \right) < 1.598123.
    \end{equation*}

    One can verify that $\E[\OPT] > 1.891332$ (detailed calculations can be found in Appendix~\ref{appendix}).
    Hence, the competitive ratio of any online algorithm for the problem is upper bounded by
    \begin{equation*}
        \E[\ALG]/\E[\OPT] \leq 1.598122/1.891332 < 0.84497.
    \end{equation*}
\end{proof}

\section{Conclusion}

In this work, we initiate the study of online stochastic matching under known independent and identically distributed (KIID) edge arrivals. We propose \textsc{Boosted Suggested Matching}, which applies \textsc{Suggested Matching} in the first phase (guided by the maximum matching of the type-graph) and \textsc{Greedy} in the second phase.
Using the powerful Natural LP as the main tool for lower bounding the matching rate of the second phase, we show that, under the assumption of integral arrival rates, this algorithm achieves a competitive ratio strictly better than $1 - 1/e$.

As the first work on online matching with KIID edge arrivals, we leave many questions open.
One of the most natural open questions is whether even higher competitive ratios than ours can be obtained under the assumption of integral rates.
We believe it is possible to derive a better lower bound on the competitive ratio of \textsc{Boosted Suggested Matching}, given a tighter analysis of the matching rate of \textsc{Greedy} in phase 2, particularly by better handling the case without perfect matchings.
It would also be interesting to explore other natural algorithms, such as the $b$-\textsc{Suggested Matching} (for example, replacing $M^*$ with a maximum $b$-matching) and investigate whether they achieve a strictly better competitive ratio.
Furthermore, note that \textsc{Suggested Matching} achieves a ratio of $1 - 1/e$ even when the graph is edge-weighted: we only need to set $M^*$ as the maximum weighted matching.
This raises a natural question: whether competitive ratios surpassing $1 - 1/e$ are possible in the edge-weighted or vertex-weighted models?
Finally, we believe it would be an interesting direction to explore the problem without the integral arrival rates assumption.
Is it possible to obtain a ratio significantly better than $0.5$ in this more general setting?

% Bibliography
\bibliographystyle{ACM-Reference-Format}
\bibliography{ref}

% Appendix
\appendix

\section{Difficulties without Integral Arrival Rates}
\label{appendix-discussion}

In this section, we briefly discuss the challenges in obtaining good competitive ratios for the KIID edge arrival model without the assumption of integral arrival rates.

Notice that the major difficulty for edge arrival models (compared with the one-sided vertex arrival model) lies in resolving the contentions among edges: for any edge $e=(i,j)$, one needs to carefully bound the probability that both $i$ and $j$ are unmatched.

Recall that with integral arrival rates, such difficulty can be (partially) circumvented by rounding the indicator solution defined by $M^*$ (as there is no contention between edges in $M^*$).
However, it does not work without integral arrival rates.
Through a similar argument as in Section~\ref{section:suggested-matching}, one would have $\Pr[e \text{ gets matched}] \geq 1-e^{-k}$ for \textsc{Suggested Matching} (when fed with the indicator solution), where $e\in M^*$ and $k$ is the lower bound on the arrival rates of edges in $M^*$.
If $k\geq 1$, $1-1/e$ is still achievable.
Nevertheless, the guarantee is meaningless if we do not have such a lower bound, which forces us to consider other fractional solutions where contentions are inevitable. 

To further demonstrate the challenges in contention resolution, we provide a $(1-e^{-2})/2 \approx 0.432$ per edge guarantee for \textsc{Suggested Matching} under KIID edge arrivals.
Consider a general fractional solution $\bx$.
Recall that \textsc{Suggested Matching} proposes each arriving edge of type $e$ with probability $x_e/\text{arrive rate of } e$.
For any edge $e=(i,j)\in E$, we have
\begin{align*}
    \frac{\Pr[e \text{ gets matched}]}{x_e} &= \frac{\sum_{t=1}^m \left(\Pr[e \text{ arrives (and is proposed) at } t] \cdot \Pr[i,j \text{ are unmatched at } t]\right) }{x_e} \\
    &\geq \sum_{t=1}^m \left( \frac{\frac{x_e}{m} \cdot \prod_{l=1}^{t-1} \left(1-\frac{\heartsuit }{m}\right)}{x_e} \right) = \frac{1}{m} \cdot \sum_{t=1}^m \left(1-\frac{\heartsuit}{m}\right)^{t-1} = \frac{1}{2} \cdot \left( 1-\left(1-\frac{\heartsuit}{m}\right)^{m} \right) \\
    &\geq \frac{1}{2} \cdot \left( 1-\left(1-\frac{2}{m}\right)^{m} \right) \approx \frac{1-e^{-2}}{2} \approx 0.432,
\end{align*}
where $\heartsuit\in [0,2)$ is the sum of $\bx$ value of all incident edges of $i$ and $j$ excluding $e$.
Notice that for a general $\bx$, $\heartsuit$ may be up to $2$ (both $i$ and $j$ are fully matched in $\bx$ without the contribution of $e$), such that the second inequality is tight.
This per-edge analysis does not imply a competitive ratio (even) higher than $0.5$, indicating that obtaining a good competitive ratio for KIID edge arrivals without integral arrival rates is highly non-trivial.

\section{Missing Proofs and Calculations} \label{appendix}

\begin{proofof}{Lemma~\ref{lemma:LP*>=E[OPT]}}
    Let $y_{e}\geq 0$ be the probability that edge $e$ belongs to the maximum matching of the realized graph.
    It suffices to prove that $\{y_{e}\}_{e\in E}$ is a feasible solution to the Natural LP.
    Notice that for any $u\in V$ and set $S \subseteq E(u)$, at most one of the edges could be included in the maximum matching of the realized graph.
    The probability that some edge in $S$ is in the maximum matching is upper bounded by the probability that at least one of these edges arrives.
    Since the online edges sample their types independently, we have
    \begin{equation*}
        \sum_{e\in S} y_{e} \leq 1-\left( 1-\frac{|S|}{m} \right)^m = 1-e^{-|S|},
    \end{equation*}
    where the equality holds since $m$ is sufficiently large.
    Therefore, $\{y_{e}\}_{e\in E}$ is feasible for the Natural LP, which completes the proof.
\end{proofof}

\begin{proofof}{Lemma~\ref{lemma:lower-bound-of-E[m2]}}
    We first define an auxiliary matching $\mathcal{M}(K,X,Y)$ as follows.
    Let $K$ be a sequence of online edges, $X$ and $Y$ be two sets of vertices.
    Let $\mathcal{M}(K,X,Y)$ be the matching produced by
    \begin{itemize}
        \item first running \textsc{Greedy} on edges $K$, where the vertices in $X$ are labeled as matched when the algorithm starts; and then 
        \item deleting all accepted edges with at least one endpoint in $Y$.
    \end{itemize}
    
    For example, if $K$ denotes the online edges arriving in the second phase, then $\mathcal{M}(K,\varnothing,\varnothing)$ is the matching produced by the unconstrained \textsc{Greedy}, and $\mathcal{M}(K,V(M_1),\varnothing)$ is the matching produced by the constrained \textsc{Greedy}, where $V(M_1)$ is the set of vertices matched in $M_1$.
    The following lemma implies a lower bound on $|\mathcal{M}(K,V(M_1),\varnothing)|$ (i.e., $m_2$).

    \begin{lemma} \label{lemma:M(X,0)>M(0,X)}
        For all sequence of edges $K$ and set of vertices $X\subseteq V$ we have
        \begin{equation*}
            |\mathcal{M}(K,X,\varnothing)| \geq |\mathcal{M}(K,\varnothing,X)|.
        \end{equation*}
    \end{lemma}
    \begin{proof}
    For convenience, let $M = \mathcal{M}(K,X,\varnothing)$ and $M' = \mathcal{M}(K,\varnothing,X)$.
    Intuitively, we have $|M|\geq |M'|$ because in $M$ we label vertices in $X$ ``unavailable'' at the \emph{beginning} of the algorithm, while in $M'$ we do the labeling at the \emph{end}: this might cause more vertices to stay unmatched in $M'$.
        
    Formally, we consider the symmetric difference between $M$ and $M'$, which is a collection of alternating cycles and paths.
    To show that $|M| \geq |M'|$, it suffices to argue that the number of edges from $M$ is at least that from $M'$ in every alternating path, since the two matchings have the same number of edges in every alternating cycle.
    Fix any alternating path, where $u$ and $v$ are the two endpoints of the path. 
    Note that $u$ and $v$ are matched in only one of $M$ and $M'$, while all other vertices are matched in both matchings.
    
    We show that the first edge $e$ being accepted in the path must be incident to $u$ or $v$. 
    Suppose otherwise, then when $e$ arrives, both its endpoints are unmatched and will be accepted in both $M$ and $M'$ when it arrives.
    Moreover, the endpoints of $e$ are not contained in $X$ and will not be removed from $M'$.
    Therefore, $e$ appears in both $M$ and $M'$, which contradicts $e$ appearing in the symmetric difference.
    Suppose edge $e$ is incident to $u$.
    We show that $e$ is matched in $M$, which implies that in the alternating path, the number of edges from $M$ is at least that from $M'$.
    Suppose $e$ is not matched in $M$, then $u$ is unmatched in $M$.
    Then, when $e$ arrives, both of its endpoints are unmatched in $M$ and should be accepted, which is a contradiction.
    Note that while $u$ is unmatched in $M'$, the above argument implies that $u$ was matched in $M'$ when $e$ arrives, but eventually becomes unmatched since it was matched to a vertex from $X$ in $\mathcal{M}(K,\varnothing,\varnothing)$.
    \end{proof}
    
    Given the above lemma, we can lower bound $\E[m_2]$ by $\E[|\mathcal{M}(K,\varnothing,V(M_1))|]$, where $K$ denotes the online edges arriving in the second phase.
    Note that the expectation is taken over both $K$ and $M_1$.
    Also, notice that the expected number of edges in $\mathcal{M}(K,\varnothing,\varnothing)$ is $\E[g((1-\rho)m)]$.
    It remains to calculate the probability $q(e)$ of an edge $e \in \mathcal{M}(K,\varnothing,\varnothing)$ being included in $\mathcal{M}(K,\varnothing,V(M_1))$.
    
    Fix any $e\in \mathcal{M}(K,\varnothing,\varnothing)$ and let $u$ and $v$ be its two endpoints.
    Since there is no edge between two endpoints in $\bar{V}$, at least one of $u$ and $v$ is matched in $M^*$.
    Note that $q(e)$ is exactly the probability that $u$ and $v$ are unmatched in $M_1$.
    
    \begin{itemize}
        \item If there exists an edge $e^*$ (which can be $e$) between $u$ and $v$ in $M^*$, then we have
        \begin{equation*}
            \textstyle q(e) = \Pr[e^* \text{ never arrives in phase 1}] = \left( 1-\frac{1}{m} \right)^{\rho\cdot m} = 
            e^{-\rho}.
        \end{equation*}

        \item Otherwise, if both $u$ and $v$ are matched in $M^*$ (say, by $e_u \in M^*$ and $e_v\in M^*$), then we have
        \begin{equation*}
            \textstyle q(e) = \Pr[\text{both } e_u \text{ and } e_v \text{ never arrive in phase 1}] = \left( 1-\frac{2}{m} \right)^{\rho\cdot m} = e^{-2\rho}.
        \end{equation*}

        \item Otherwise, suppose only $u$ is matched in $M^*$ (say, by $e_u\in M^*$).
        Notice that $v$ cannot be matched in Phase 1.
        Therefore we have
        \begin{equation*}
            \textstyle q(e) = \Pr[e_u \text{ never arrives in phase 1}] = \left( 1-\frac{1}{m} \right)^{\rho\cdot m} = 
            e^{-\rho}.
        \end{equation*}
    \end{itemize}

    Therefore, we have
    \begin{align*}
        &\ \E[|\mathcal{M}(K,\varnothing,V(M_1))|] \\
        \geq &\ e^{-2\rho}\cdot \E[|\mathcal{M}(K,\varnothing,\varnothing)\cap (\hat{E}\setminus M^*)|] + e^{-\rho}\cdot \E[|\mathcal{M}(K,\varnothing,\varnothing)\cap (\bar{E}\cup M^*)|
        \\
        =&\ e^{-2\rho}\cdot \E[\hat{g}((1-\rho)m)] + e^{-\rho}\cdot \E[\bar{g}((1-\rho)m)].
    \end{align*}
    which completes the proof.
\end{proofof}

\begin{proofof}{Lemma~\ref{lemma:lower-bound-of-E[g]}}
    Recall that by definition $\E[w(0)] = n\cdot f(0) = 0$.
    Notice that $b(z)$ is decreasing in $z\in [0,1)$ and $b(z) > 0$ for all $z < (1-\epsilon)/2$.
    For $y\geq 0$ such that $b(y)\geq 0$, define
    \begin{equation*}
        B(y) = \int_0^y \frac{1}{b(z)}\ \d z.
    \end{equation*}

    For all $y>0$ we have
    \begin{equation*}
        B\left(f(y)\right) = \int_0^{f(y)} \frac{1}{b\left(f(x)\right)} \d f(x) = \int_0^{f(y)} \frac{1}{f'(x)} \d f(x) = \int_0^y \d x = y.
    \end{equation*}
    
    On the other hand, we have
    \begin{align*}
        B\left(\frac{\E[w(k)]}{n}\right) & = \int_0^{\frac{\E[w(k)]}{n}} \frac{1}{b(z)}\ \d z = \sum_{i=0}^{k-1} \int_{\frac{\E[w(i)]}{n}}^{\frac{\E[w(i+1)]}{n}} \frac{1}{b(z)}\ \d z \\
        & \geq \sum_{i=0}^{k-1} \int_{\frac{\E[w(i)]}{n}}^{\frac{\E[w(i+1)]}{n}} \frac{1}{b\left(\frac{\E[w(i)]}{n}\right)}\ \d z \\
        & = \sum_{i=0}^{k-1} \frac{1}{b\left(\frac{\E[w(i)]}{n}\right)}\cdot \left( \frac{\E[w(i+1)] - \E[w(i)]}{n} \right) \\
        & \geq \sum_{i=0}^{k-1} \frac{1}{m} = \frac{k}{m} = B\left( f\left( \frac{k}{m} \right) \right),
    \end{align*}
    where the first inequality holds since $b(z)$ is decreasing; the second inequality holds due to Inequality~\eqref{equation:lower-bound-on-increment-of-E[g(k)]}.
    Since $B(y)$ is increasing, we have $\E[w(k)]\geq n\cdot f(\frac{k}{m})$.
\end{proofof}

\begin{proofof}{Lemma~\ref{lemma:large_expected_g}}
    As shown in Section~\ref{section:Analysis_framework}, 
    given the realization of the first $k$ rounds, we have
    \begin{equation*}
        \E_{k+1}[w(k+1)] = w(k) + \frac{\Delta(k)}{m},
    \end{equation*}
    where the expectation is taken over the randomness of round $k+1$.
    Lemma~\ref{lemma:small_w_large_delta} states if $w(k)\leq \epsilon'\cdot n$ then the expected increase $w(k+1)-w(k)$ in the next round is lower bounded by $\frac{n}{m}\cdot c(\rho,\epsilon,\epsilon')$.

    Observe that for all $\delta > 0$ we have
    \begin{align*}
        \E[w(k^*)] & \geq (1-\delta) \epsilon' n\cdot \Pr\left[ w(k^*) \geq (1-\delta) \epsilon' n \right] \\
        & \geq (1-\delta) \epsilon' n\cdot \Big( 1 - \Pr\left[ w(k^*) < (1-\delta) \epsilon' n \right] \Big).
    \end{align*}

    Next, we show that the event $w(k^*) < (1-\delta) \epsilon' n$ happens with very low probability.

    Let $z(k) = w(k+1)-w(k)$.
    Note that $z(0),z(1),\ldots,z(k^*-1)$ are non-negative random variables with values bounded by $e$, and we have $w(k) = \sum_{i=0}^{k-1} z(i)$.
    In the following, we would like to give a lower bound for
    \begin{equation*}
        \Pr\left[ w(k^*) < (1-\delta) \epsilon' n \right] = \Pr\left[ \sum_{i=0}^{k^*-1}z(i) < (1-\delta) \epsilon' n \right].
    \end{equation*}
    Note that event $w(k^*) < (1-\delta) \epsilon' n$ implies $w(k) <\epsilon'\cdot n$ for all $k\leq k^*$.
    We further define $k^*$ IID random variables as follows.
    For each $i\in \{0,1,\ldots,k^*-1\}$, let
    \begin{equation*}
        \bar{z}(i) = \begin{cases}
            0, \text{ with probability } 1-\frac{n\cdot c(\rho,\epsilon,\epsilon')}{e\cdot m}; \\
            e, \text{ with probability } \frac{n\cdot c(\rho,\epsilon,\epsilon')}{e\cdot m}.
        \end{cases}
    \end{equation*}
    We have $\E[\bar{z}(i)] = \frac{n}{m}\cdot c(\rho, \epsilon, \epsilon')$ by definition.

    Then we have
    \begin{align*}
        \Pr\left[ w(k^*) < (1-\delta) \epsilon' n \right]
        &= \Pr\left[ \sum_{i=0}^{k^*-1} z(i) < (1-\delta) \epsilon' n \right]\\
        &\leq \Pr\left[ \sum_{i=0}^{k^*-1} \bar{z}(i) < (1-\delta) \epsilon' n \right]\\
        &= \Pr\left[\sum_{i=0}^{k^*-1} \bar{z}(i) < \E\left[ \sum_{i=0}^{k^*-1}\bar{z}(i) \right] - \delta \epsilon' n \right]\\
        &\leq \frac{\sum_{i=0}^{k^*-1}\text{Var}(\bar{z}(i))}{(\delta\cdot \epsilon'\cdot n)^2} \\ & \leq \frac{k^*\cdot e \cdot n/m\cdot c(\rho,\epsilon,\epsilon')}{(\delta\cdot \epsilon'\cdot n)^2}=\frac{e}{\delta^2\cdot \epsilon'\cdot n},
    \end{align*}
    where the second inequality holds by Chebyshev's inequality and the last inequality holds by Bhatia-Davis inequality.
    Finally, by setting $\delta = 1/\sqrt[3]{n}$, we have
    \begin{equation*}
        \E[w(k^*)] \geq \left( 1-\frac{1}{\sqrt[3]{n}} \right) \epsilon' n\cdot \left( 1 - \frac{e}{\epsilon'\cdot \sqrt[3]{n}} \right) = (1-o(1))\cdot \epsilon'\cdot n,
    \end{equation*}
    which completes the proof.
\end{proofof}

\vspace*{5pt} 
\textbf{Calculation of $E[\OPT]$ in the Hard Instance: }
    % Next, we compute $\E[\OPT]$.
    Observe that we have $\OPT = 1$ if and only if one of the following two events happens:
    \begin{itemize}
        \item $A_1$: All online edges share a common endpoint.
        \item $A_2$: The online edges form a triangle.
    \end{itemize}

    Note that
    \begin{align*} 
        \Pr[A_1] & = \sum_{u\in V} \Pr[\text{all edges share vertex } u] - \sum_{e\in E} \Pr[\text{only } e \text{ realizes}] \\
        & = 4\cdot (3/6)^6 - 6\cdot (1/6)^6 = 485/7776.
    \end{align*}

    Let $T\subseteq E$ be the set of all triangles on $G$.
    Note that $|T| = 4$.
    The probability that the online edges form a triangle is
    \begin{align*}
        \Pr[A_2] & = \sum_{T\subseteq E} \bigg( \Pr[\text{all online edges are sampled in } T] \cdot \\
        & \qquad \Pr[\text{each $e$ in $T$ appears at least once} \mid \text{all edges are sampled in } T] \bigg) \\
        & = 4\cdot \left( (3/6)^6 \cdot \frac{3^6 - 3\cdot 2^6 + 3\cdot 1^6}{3^6} \right) = \frac{5}{108}.
    \end{align*}

    Further, notice that events $A_1$ and $A_2$ are disjoint, we have
    \begin{equation*}
        \Pr[\OPT=1] = \Pr[A_1] + \Pr[A_2] = 485/7776 + 5/108 = 845/7776.
    \end{equation*}
    
    Therefore, we can derive $\E[\OPT]$ as
    \begin{equation*}
        \E[\OPT] = \Pr[\OPT=1] + 2\cdot \Pr[\OPT=2] = 845/7776 + 2\cdot (1 - 845/7776) > 1.891332.
    \end{equation*}

\end{document}